\documentclass[letterpaper,twocolumn,10pt]{article}
\usepackage{usenix2019_v3}
\usepackage{pgfplots}
\usepackage{mathtools}

\usepgfplotslibrary{units,groupplots}

\usepackage{adjustbox}
\usepackage{setspace}
\usepackage{tikz}
\usetikzlibrary{pgfplots.groupplots, matrix,fit}
\usetikzlibrary{calc}
\usepackage{amsmath}
\usepackage{amsthm, thmtools, thm-restate}
\usepackage{amsfonts}
\usepackage{siunitx}
\usepackage{subcaption}
\usepackage{multirow}

\DeclareMathAlphabet{\altmathcal}{OMS}{cmsy}{m}{n}

\usepackage{booktabs}
\newtheorem{theorem}{Theorem}

\makeatletter
\renewcommand{\paragraph}{%
  \@startsection{paragraph}{4}%
  {\z@}{1.2ex \@plus 1ex \@minus .2ex}{-1em}%
  {\normalfont\normalsize\bfseries}%
}
\makeatother

\newcommand{\smidge}{{\kern .05em}}
\newcommand{\Colon}{\smidge\colon\smidge}
\newcommand{\bits}{\{0,1\}}
\newcommand{\Z}{\mathbb{Z}}
\newcommand{\R}{\mathbb{R}}

\DeclareMathOperator{\prob}{\textnormal{Pr}}
\DeclareMathOperator{\e}{E}

\newcommand{\getsr}{{\;{\leftarrow{\hspace*{-3pt}\raisebox{.75pt}{$\scriptscriptstyle\$$}}}\;}}
\newcommand{\getdist}[1]{{\;{\leftarrow{\hspace*{-3pt}\raisebox{.75pt}{$\scriptscriptstyle #1$}}}\;}}

\newcommand{\Prob}[1]{\prob\left[#1\right]}
\newcommand{\CondProb}[2]{\prob\left[#1\;|\;#2\right]}
\newcommand{\Ex}[1]{\textnormal{E}\left[#1\right]}
\newcommand{\epsilonPrecR}[1]{\epsilon_{prec}^{#1}}

\newcommand{\advA}{\altmathcal{A}}

\newcommand{\false}{\textsf{false}}

\newcommand{\Sketch}{SS}
\newcommand{\Rec}{Rec}
\newcommand{\sketchval}{z}
\newcommand{\minent}{\mu}
\newcommand{\redminent}{\mu'}

\newcommand{\G}{\mathbb{G}}

\newcommand{\minentropy}{\mathbf{\mathsf{\tilde{H}}}_\infty}
\newcommand{\condminentropy}[2]{\minentropy({#1}\,|\,{#2})}
\newcommand{\prfval}{Y}
\newcommand{\blindval}{X}
\newcommand{\blindprfval}{Y'}

\newcommand{\hashspace}{U}
\newcommand{\inputspace}{\mathbb{I}}
\newcommand{\len}{\ell}
\newcommand{\emblen}{d}
\newcommand{\embed}{\simhash}
\newcommand{\simhash}{\altmathcal{F}}
\newcommand{\simimage}{v}
\newcommand{\imagespace}{\altmathcal{W}}
\newcommand{\simimagespace}{\altmathcal{V}}
\newcommand{\metricspace}{\altmathcal{Y}}

\newcommand{\distimage}{\Delta_\imagespace}
\newcommand{\image}{w}
\newcommand{\condition}{cond}
\newcommand{\embimage}{p}
\newcommand{\imageset}{\altmathcal{B}}
\newcommand{\bucketset}{B}

\newcommand{\threshold}{T}
\newcommand{\coarsethreshold}{k}
\newcommand{\flipbias}{\gamma}
\newcommand{\repeats}{q}

\newcommand{\flip}{\textnormal{Flip}}
\newcommand{\pnoise}{\tilde{p}}
\newcommand{\vnoise}{\tilde{v}}
\newcommand{\indexfunc}{I}
\newcommand{\indexfuncset}{\altmathcal{I}}

\newcommand{\baseline}{\epsilon_{base}}
\newcommand{\probThreshold}{T_{adv}}

\newcommand{\advAprec}{\advA_{\textrm{pre}}}

\newcommand{\AUC}{\textnormal{AUC}}
\newcommand{\advAauc}{\advA_{\textrm{auc}}}

\newcommand{\PRED}{\textnormal{PRED}}
\newcommand{\epsilonAcc}{\epsilon_{\mathrm{acc}}}
\newcommand{\epsilonAuc}{\epsilon_{\mathrm{auc}}}
\newcommand{\epsilonRecall}{r}
\newcommand{\epsilonRecallThresh}{\rho}
\newcommand{\advAacc}{\advA_{\textrm{acc}}}

\newcommand{\true}{\mathsf{true}}

\newcommand{\EmbedScheme}{\mathbb{E}}
\newcommand{\Embed}{\textnormal{Emb}}
\newcommand{\Sim}{\textnormal{Sim}}

\newcommand{\dist}{\mathcal{D}}
\newcommand{\distf}{\dist_\simhash}

\newcommand{\setting}{\pi}

\newcommand{\vecimage}{\mathbf{\image}}
\newcommand{\vecembimage}{\mathbf{\embimage}}

\newcommand{\myInd}{\hspace*{1em}}

\newcommand{\bnm}{\begin{newmath}}
\newcommand{\enm}{\end{newmath}}
\newcommand{\bne}{\begin{newequation}}
\newcommand{\ene}{\end{newequation}}

\newenvironment{newmath}{\begin{displaymath}%
\setlength{\abovedisplayskip}{4pt}%
\setlength{\belowdisplayskip}{4pt}%
\setlength{\abovedisplayshortskip}{6pt}%
\setlength{\belowdisplayshortskip}{6pt} }{\end{displaymath}}

\newenvironment{newequation}{\begin{equation}%
\setlength{\abovedisplayskip}{4pt}%
\setlength{\belowdisplayskip}{4pt}%
\setlength{\abovedisplayshortskip}{6pt}%
\setlength{\belowdisplayshortskip}{6pt} }{\end{equation}}

\newcommand{\secref}[1]{Section~\ref{#1}}
\newcommand{\apref}[1]{Appendix~\ref{#1}}
\newcommand{\figref}[1]{Figure~\ref{#1}}
\newcommand{\tabref}[1]{Table~\ref{#1}}

\newcommand{\verylongrightarrow}[1]             
      {\setlength{\unitlength}{.01in}           
      \begin{picture}(#1,1) \put(0,0){\vector(1,0){#1}} \end{picture}}

\newcommand{\gamesfontsize}{\footnotesize}
\newcommand{\stretchval}{1.2}

\newcommand{\fpage}[2]{\framebox{\begin{minipage}{#1\textwidth}\setstretch{\stretchval}\gamesfontsize #2 \end{minipage}}}

\newcommand{\hfpagess}[4]{
		\begin{tabular}{c@{\hspace*{.5em}}c}
		\framebox{\begin{minipage}[t]{#1\textwidth}\setstretch{\stretchval}\gamesfontsize #3 \end{minipage}}
		&
		\framebox{\begin{minipage}[t]{#2\textwidth}\setstretch{\stretchval}\gamesfontsize #4 \end{minipage}}
		\end{tabular}
	}

\def\codestretch{\stretchval}

\newcommand{\hpagess}[4]{
    \begin{tabular}[t]{c@{\hspace*{1.5em}}c}
	   \adjustbox{valign=c}{\begin{minipage}[t]{#1\textwidth}\setstretch{\codestretch} #3 \end{minipage}}
	   &
     \adjustbox{valign=c}{\begin{minipage}[t]{#2\textwidth}\setstretch{\codestretch} #4 \end{minipage}}
    \end{tabular}
	}

\newlength{\saveparindent}
\newlength{\saveparskip}
\newcounter{ctr}

\newenvironment{newitemize}{%
\begin{list}{\mbox{}\hspace{5pt}$\bullet$\hfill}{\labelwidth=15pt%
\labelsep=5pt \leftmargin=20pt \topsep=3pt%
\setlength{\listparindent}{\saveparindent}%
\setlength{\parsep}{\saveparskip}%
\setlength{\itemsep}{3pt} }}{\end{list}}

\definecolor{clr3_1}{RGB}{203,106,73}
\definecolor{clr3_2}{RGB}{164,108,183}
\definecolor{clr3_3}{RGB}{122,164,87}

\definecolor{clr5_1}{RGB}{75,174,141}
\definecolor{clr5_2}{RGB}{202,86,136}
\definecolor{clr5_3}{RGB}{133,160,64}
\definecolor{clr5_4}{RGB}{135,116,202}
\definecolor{clr5_5}{RGB}{202,112,64}

\definecolor{clr2_1}{RGB}{179,102,158}
\definecolor{clr2_2}{RGB}{152,152,77}

\definecolor{clr6_1}{RGB}{167,221,226}
\definecolor{clr6_2}{RGB}{230,184,179}
\definecolor{clr6_3}{RGB}{155,194,175}
\definecolor{clr6_4}{RGB}{209,187,223}
\definecolor{clr6_5}{RGB}{212,217,182}
\definecolor{clr6_6}{RGB}{170,196,226}

\definecolor{clr6_1}{RGB}{121,113,168}
\definecolor{clr6_2}{RGB}{121,177,69}
\definecolor{clr6_3}{RGB}{154,72,190}
\definecolor{clr6_4}{RGB}{89,141,108}
\definecolor{clr6_5}{RGB}{183,73,89}
\definecolor{clr6_6}{RGB}{185,124,63}

\begin{document}

\date{}

\title{Increasing Adversarial Uncertainty to Scale Private Similarity Testing}

\author{
  {\rm Yiqing Hua$^{1,2}$, Armin Namavari$^{1,2}$, Kaishuo Cheng$^{2}$, Mor
  Naaman$^{1,2}$, Thomas Ristenpart$^{1,2}$}\\
  $^{1}$ Cornell Tech \hspace*{3em} $^{2}$ Cornell University
} 

\maketitle

\begin{abstract}

Social media and other platforms rely on automated detection of abusive content
to help combat disinformation, harassment, and abuse.  
One common approach is to check user content for similarity against a server-side database of problematic items.
However, this method fundamentally endangers user privacy.
Instead, we target client-side detection, 
notifying only the users when such matches occur to warn them against abusive content.

Our solution is based on privacy-preserving similarity testing. 
Existing approaches rely on expensive
cryptographic protocols that do not scale well to large databases and may
sacrifice the correctness of the matching.  
To contend with this challenge, we propose and formalize the
concept of similarity-based bucketization~(SBB). With SBB, a client reveals a
small amount of information to a database-holding server so that it can generate
a bucket of potentially similar items.  
The bucket is small enough for efficient
application of privacy-preserving protocols for similarity.  
To analyze the
privacy risk of the revealed information, we introduce a framework
for measuring an adversary's confidence in inferring a predicate about the client
input correctly.  
We develop a practical SBB protocol for image content, and evaluate its client
privacy guarantee with real-world social media data. We then combine SBB
with various similarity protocols, showing that the combination with SBB provides a speedup of
at least $29\times$ on large-scale databases compared to that without, while retaining correctness of over $95\%$.

\end{abstract}

\section{Introduction}

Faced with various policy-violating activities ranging from
disinformation~\cite{resende2019mis}
to harassment~\cite{hua2020characterizing,matias2015reporting,duggan2014online} and abuse~\cite{bursztein2019rethinking,revenge}, social media companies increasingly rely
on automated algorithms to detect deleterious content.  One
widely used approach is to check that user content is not too similar to
known-bad content.
For example, to detect child sexual abuse
imagery~\cite{bursztein2019rethinking},
some platforms utilize similarity hashing
approaches like PhotoDNA~\cite{photodna} or PDQHash~\cite{pdqhash}.
These approaches map
user-shared images into unique representations that encode perceptual structure,
enabling quick comparisons against a database of hash values.
Such approaches could be helpful for combating other forms of bad content, such
as the viral spread of visual
misinformation on end-to-end encrypted messaging services~\cite{resende2019mis}.
For example, they could augment other efforts to provide users with important context about shared content~\cite{googlefact,whatsappsearch}.

Currently deployed approaches rely on sending user content or a similarity hash of the
content to a moderation service. This risks user privacy. As we detail in the
body, the service can easily match a submitted similarity hash against known
images to learn the content of a user's image with overwhelming confidence.
Privacy can be improved utilizing cryptographic two-party computation
(2PC)~\cite{yao1986generate,kulshrestha2021identifying} techniques to only
reveal matching content to the moderation service and nothing more. The recent
CSAM image detection system proposed by Apple~\cite{apple_csam} goes one step
further and notifies the platform only when the number of matching images 
surpasses a certain threshold.  Automated notification of platforms 
necessarily raises concerns about privacy and accountability (e.g., how to
ensure the system is not used for privacy-invasive search for benign images).

An alternative approach is to have only the client learn the output of similarity
checks, to enable client-side notifications, warning or otherwise informing users.
This may not be suitable for all classes of abusive content, such as CSAM, where
the recipient may be adversarial, but could be useful for other abuse categories
(misinformation, harassment, etc.). 
However, the scale of databases makes it prohibitive both to send 
the known-bad hashes to the client or, should hashes be sensitive, apply 2PC
techniques to ensure as little as possible about the database leaks to clients. 
For an example of the latter, Kulshrestha and Mayer's~\cite{kulshrestha2021identifying} 
private approximate membership computation (PAMC) protocol achieves 
state-of-the-art performance, but nevertheless requires 
about~$27$ seconds to perform a similarity check
against a database with one million images.
The protocol also has an average false negative rate of almost $17\%$ for
slightly transformed images, meaning many similar images may be erroneously
marked as dissimilar.

In this work, 
we target client-side detection,
in order to warn users against abusive content.
To this end, we explore the question of how to scale privacy-preserving image
similarity protocols, while preserving correctness of the similarity testing.
We introduce and formalize the concept of 
similarity-based bucketization (SBB).  The idea is to reveal a small amount of
structured information in a message to a database-holding server, so that it can determine a
bucket of possibly relevant database entries. Ideally the bucket consists of only a 
small fraction of the full database, enabling use of a subsequent similarity
testing protocol on the bucket to perform the final similarity check. We explore
instantiating the testing protocol in a variety of ways. 

The key technical challenge facing SBB is balancing the competing goals of 
minimizing bucket size (efficiency) with leaking as little information as
possible (privacy). For example, one could modify a standard similarity hash,
say PDQHash, to provide only very coarse comparisons. But as we will show, this still leaks
a lot of information to the server, allowing high-confidence attacks
that can associate the coarse hash to the specific content of a client request. More broadly we need a framework
for navigating this tension. 

We propose such a framework.  It formalizes information leakage using a
game-based definition. 
To be specific,
an adversarial server attempts to learn, from an
SBB message generated for some image drawn from an adversary-known
distribution, a predicate about the underlying image. As an important
running example, we use a ``matching predicate'' that checks if 
the underlying image has the same perceptual hash value as that of a known target image. 
Unlike in more traditional cryptographic definitions
(e.g.,~\cite{goldwasser1984probabilistic}), we do not require the adversarial
server to have negligible success (which would preclude efficiency)
and instead offer a range of measures including accuracy
improvement over baseline guessing, adversarial precision, and adversarial
area under the receiver operating characteristic curve~(AUC).  Indeed, there is no one-size-fits-all approach to
measuring privacy damage, and our framework allows one to more broadly
assess risks.

We offer a concrete SBB mechanism that increases adversarial uncertainty
compared to naive approaches. It converts any similarity hash that uses
Hamming distance to a privacy-preserving coarse embedding; we focus on PDQHash
because it is widely supported.  We combine techniques from 
locality-sensitive hashing~\cite{gionis1999similarity} with lightweight noise
mechanisms.  The ultimate algorithm is conveniently simple: apply a standard
PDQHash to an image, choose a designated number~$\emblen$ of bit indices randomly,
flip each selected bit with probability~$\flipbias$, and then send the
resulting~$\emblen$ bits and their indices to the server. An image in the server's database is
included in a bucket should~$\coarsethreshold$ or fewer of the relevant~$\emblen$ bits of its PDQHash
mismatch with those that are sent from the client.  

Using real-world social media data, we empirically assess correctness, efficiency and privacy under various definitions.
We explore various settings of $\emblen$, $\flipbias$, and $\coarsethreshold$, and show that it is possible to ensure average bucket sizes
of~$9.3\%$ of the database,
while: (1)~ensuring that the similar
images are included in the bucket at least 95\% of the time, and (2) an
optimal adversary for the matching predicate achieves less than 50\% precision,
signifying low confidence in matching attacks. We caution that these empirical
results are dataset-dependent, and may not generalize to every potential use case.
Instead they can be interpreted as a proof-of-concept that SBB works in
a realistic scenario.

We then combine our SBB mechanism with various similarity protocols, 
with different privacy guarantees for the server's content.
For the expedient approach of downloading the bucket of server PDQHash values and
performing comparisons on the client side,
SBB provides a speedup of $29\times$ or more.
A full similarity check requires less than $0.5$ seconds for a database of $2^{23}$ images.
We also explore using SBB to speed up an ad hoc similarity protocol based
on secure sketches~\cite{dodis2004fuzzy}, as well as 2PC protocols implemented
in the EMP~\cite{emp-toolkit} and CrypTen~\cite{crypten2020} frameworks. Our experiments
indicate that SBB can provide
speed-ups of $601\times$, $97\times$, and $67\times$,
respectively, and often enables use of 2PC that would fail otherwise
due to the size of the database.

We conclude by discussing various limitations of our results, and open questions
that future work should answer before deployment in practice.
Nevertheless, we expect that our SBB approach will be
useful in a variety of contexts. Encrypted messaging apps 
could use it to help warn users about malicious content, with significantly better privacy
than approaches that send
plaintext data to third-party servers~\cite{googlefact,whatsappsearch}.  
In another setting, social media platforms that currently query their users'
plaintext data to third-party
services to help identify abuse (e.g.,~\cite{threatexchange,joint}) could use
our techniques to improve privacy for their users. 
To facilitate future research, 
our prototype implementation is publicly available.\footnote{\url{https://github.com/vegetable68/sbb}}

\section{Background and Existing Approaches}
\label{sec:overview}

In this section, we provide some background about a key motivating
setting: providing client-side detection of bad content in end-to-end (E2E)
encrypted messaging. That said, our approaches are more general and we 
discuss other deployment scenarios in \apref{ap:mod}.

\paragraph{Content detection and end-to-end encryption.}
Content moderation aims to mitigate abuse on social media platforms, and can
include content removal, content warnings, blocking users, and more.
Most moderation approaches rely on detecting objectionable content, 
particularly at scale where automated
techniques seem to be requisite. Social
media companies often maintain large databases of known adversarial
content~\cite{bursztein2019rethinking,photodna} and compare a client message
with items in the databases to see if the message is sufficiently similar to
some piece of adversarial content.  However,
this approach requires the client to reveal plaintext message content, which
stands in tension with privacy-preserving technologies like E2E
encryption.  On the other hand, leaving contents unmoderated on the platform is
unsatisfactory given the harms caused by abusive content such as misinformation,
child sexual abuse material (CSAM), harassment, and more.

Governments\footnote{\url{https://www.justice.gov/opa/press-release/file/1207081}}
and non-governmental
organizations\footnote{\url{https://www.missingkids.org/e2ee}} have for many
years emphasized the need for  technical innovations that could enable law
enforcement access to encrypted data,  while minimizing risks of privacy
violation~\cite{rozenshtein_child,group_moving,eu_report}.  However, security
experts have repeatedly expressed concern that such `backdoor' access would
fundamentally break the privacy of E2E
encryption~\cite{cdt_new,crodker_dont,portnoy_why,muffet_what} or, if it
provided content blocking functionality, enable problematic censorship~\cite{portnoy_why}.  

In this work, we target mitigations that allow privacy-preserving client-side
detection of content similar to known bad content. We focus on images, as
discussed below. Our protocol is agnostic to how client software uses this
detection capability, 
but we believe that client software should be designed to 
empower users with information and the ability to make their own
decisions about content.

Our techniques may be useful, for example, to mitigate the
increasing use of E2E encrypted messaging for harmful disinformation campaigns~\cite{resende2019mis,gursky2021countering}. 
A widely discussed approach is to warn users against known disinformation.
Recent research~\cite{kaiser2021adapting} has shown that 
when carefully designed,
such warnings are effective in guiding user behaviors to avoid disinformation.
Our work provides a technical solution for the client-side warning mechanism.
To be specific, the proposed system queries whether a client's received content is similar to
known disinformation and returns the answer only to the client.
Such a design avoids both outright censorship and notifying platform operators that a particular client received a
particular piece of content. 
This solution would enable the kinds of user-initiated known content detection
approaches that have been suggested
recently~\cite{mayer_content,callas_thoughts}, and could help complement
existing anti-abuse techniques that do not consider content, such as those used
in WhatsApp~\cite{whatsapp_stop}. 

But warning-style approaches that inform and empower users may not 
be suitable for threats like CSAM, where the recipients of messages can themselves
be bad actors.  Here client software would seemingly have to limit
user choice, automatically blocking detected content and/or notifying some
authority about it.
Recent designs for CSAM mitigations include the Kulshrestha-Mayer
protocol~\cite{kulshrestha2021identifying} (when used to notify the platform)
and the CSAM detection
proposal by Apple~\cite{apple_csam}.
Cryptographers have, in turn, raised the alarm that, while efforts to combat
CSAM are laudable, these platform-notifying systems represent a potential
E2E encryption backdoor that is subject to misuse by platform operators or
governments~\cite{mckinney_apple,green_apple} and that 
future work is needed to make such systems transparent and accountable.
Our work is different, as we target client-side notification and not platform
notification.

Another concern is that even client-side notification ends up a stepping stone
towards riskier backdoor/censorship mechanisms, because once the former is
deployed it will be easier to deploy, or justify deploying, the latter.
Client-side functionality at least provides the opportunity for activists and
others to detect changes to client-side software and understand their effects,
adding some transparency and accountability.  At the same time, arguments for,
or against, various anti-abuse mechanisms would do well to delineate between
approaches that empower users to understand and control their online experience
(warnings, the ability to select users/content to block) and that disempower
users (client-side or platform-side automatic censorship).  We
believe our techniques will be useful for the former, without intrinsically
promoting the latter.

\paragraph{Client-side similarity testing and privacy.}
\label{sec:deployment}
As mentioned above, we focus on private image similarity testing services. These
allow a client,
who receives some value $\image$ on an E2E encrypted platform,
to submit a request to a service provider holding a database $\imageset$; the response indicates to the client whether
$\image$ is
similar to any item in $\imageset$. As the database $\imageset$ may be quite
large, we need scalable solutions.
The service provider could be the messaging platform,
or a third party service.
In the case when the provider is a third party service,
the protocol runs between the client and the testing service, without
involvement of platform servers.

A key concern will be the privacy risk imposed on clients by a testing service.
Our threat model consists of an adversary in control of the service's servers,
who wants to learn information about a client's image~$\image$ by inspecting
messages sent to the service in the course of similarity testing. This is often
referred to as a semi-honest adversary, though our approaches will
meaningfully resist some types of malicious adversaries that deviate
from the prescribed protocol.
In terms of privacy threat, we primarily focus on what we call a matching attack, in which the
adversary wishes to accurately check whether $\image$ matches some adversarially
chosen image (see \secref{sec:privacy-goal} for a formalization).  A matching
attack enables, for example, adversarial service operators to monitor whether 
clients received any image on an adversarially chosen watchlist.  

In this initial work we primarily focus on the risks against a single query from
the client, and explicitly do not consider adversaries that just
want to recover partial plaintext information, such as if the adversary wants to
infer if an image contains a person or not. While we believe our results also
improve privacy for such attacks, we do not offer evidence either way and future
work will be needed to explore such threats. 
We also do not consider
misbehaving servers that seek to undermine correctness, e.g., by modifying
$\imageset$ to force clients to erroneously flag innocuous images.  How to build
accountability mechanisms for this setting is an interesting open question.
We simulate the scenarios of adversaries that somehow can take
advantage of known correlations between queried images in \apref{ap:repeated} 
and propose potential mitigation solutions in \secref{sec:limitations}.

Nevertheless there are already several challenges facing developing a service
that prevents accurate matching attacks in our setting.
While prior work has established
practical protocols for private set
membership~\cite{thomas2019protecting,li2019protocols}, these only provide exact
equality checks.  Even small manipulations such as image re-sampling, minor
cropping, or format conversion make exact matching schemes fail.  Second,  the
database $\imageset$ can be arbitrarily large and may require frequent update.
For instance, the published dataset from Twitter with activities of accounts
associated to the
Russian Internet Research Agency consist of 2 million images in
total~\cite{twitter_ira}.

\paragraph{Existing approaches.}
We review deployed systems and suggested
designs for image similarity testing.

\textbf{\emph{Plaintext services.}} Most current deployments have the client 
upload their image to a third party service. 
A prominent example is the PhotoDNA
service.
After a client submits an image to the service, it immediately hashes the image using a 
proprietary algorithm~\cite{photodna}. Importantly, the hash can be compared to other
hashes of images in a way that measures similarity of the original images. Such
hashes are often called similarity hashes~\cite{oliva2001modeling,chum2008near} or perceptual hashes~\cite{zauner2010implementation}. (We show 
examples later.)
The original image that was sent to the service is deleted after hashing.
This plaintext design has various benefits, including simplicity for clients
and the ability to hide the details of the hashing used. The
latter is important in contexts where malicious users attempt to
modify an image $\image \in \imageset$  in the service's bad list 
to create an image $\image'$ that will not be found as similar to any image in
$\imageset$ (including $\image$)~\cite{xiao2019seeing}. 

Another example of a plaintext service is WhatsApp's in-app reverse search
function to combat visual misinformation~\cite{whatsappsearch},
rolled out in June 2020.  This feature
allows users to submit their images to Google reverse image search for the
source or context of a specific image.  In this case, the user needs to reveal
their image to both Google and 
WhatsApp, 
sacrificing user privacy.

\textbf{\emph{Hashing-based services.}}
For privacy-aware clients,
revealing plaintext images represents a significant privacy risk. 
An alternative approach is to use a public hashing algorithm,
have the client first hash their image,
and submit only the resulting representation to the similarity checking service.
While this requires making the hashing algorithm available to clients (and,
potentially, adversarial users), it improves privacy because the original images are 
not revealed to the service. It also improves performance: hashes can be
compact (e.g., 256 bits) and compared against a large database
$\imageset$ in sublinear time~\cite{norouzi2012fast}. This approach is used by
Facebook's ThreatExchange~\cite{threatexchange} service that allows
organizations to share hashes of images across trust boundaries. 
They use a custom similarity
hash called PDQHash~\cite{pdqhash}.

Sharing hashes, however, still has privacy risk. 
For example, 
although the lossy process of PDQHash generation makes recovering the exact input impossible in general,
revealing the hash allows inferring whether a queried value
is similar to another image.
An adversary at the service provider's side may brute-force search a database of 
images to find ones close to the queried value. 

\textbf{\emph{Cryptography-based services.}} An alternative approach that preserves
privacy is to employ a secure 2PC protocol~\cite{yao1986generate}
between the client and service. 
Existing 2PC protocols for similarity matching
(e.g.,~\cite{jarrous2013secure,asharov2018privacy,chen2020sanns}) can,  
in the best case, ensure that no information about the client's image is
leaked to the server and that nothing about $\imageset$ (beyond whether it
contains a similar image) leaks to the client.
However, existing 2PC protocols do not efficiently scale to large 
databases~$\imageset$.

Recent work by Kulshrestha and Mayer~\cite{kulshrestha2021identifying} proposed private approximate membership computation~(PAMC)
to allow similarity testing of images encoded as PDQHashes.
The protocol begins by splitting the database $\imageset$ into buckets.
Using private information retrieval, a client retrieves a bucket from the server with the bucket identifier generated from the PDQHash of their image.
The chosen bucket is not disclosed to the server.
The two parties then perform a private similarity test 
to determine whether the client PDQHash has sufficiently small Hamming
distance to any image in the bucket.
The protocol is still rather expensive, with their initial experiments requiring $37.2$ seconds for a one-time set up 
and $27.5$ seconds for a query for a block list with the size of $2^{20}$. These
times exclude network delays (measurements were performed with client and server
on the same workstation).
While a step closer to practicality, this remains prohibitive particularly since
we expect that performance in deployment would be worse for lightweight client
hardware such as mobile phones. 

Concurrent work by Apple~\cite{apple_csam} proposed a protocol that encodes images from a user's cloud storage into perceptual hashes.
The perceptual hashing algorithm maps similar images into identical hashes with high probability.
The protocol then performs private set intersection between the encoded hashes and a database of known CSAM images.
The private contents are revealed to the platform only when the number of matches exceeds a certain threshold.
Whether such a protocol, designed for CSAM detection, is fit for client-side detection remains a question for future work to explore.

In summary, all three existing design approaches for image similarity 
testing ---
revealing images as client requests,
using similarity representations like PDQHash as client requests,
and employing secure 2PC protocols --- do not provide satisfying solutions. 
The first two designs do not provide sufficient
privacy, while 2PC designs are currently not sufficiently efficient.
Thus we need a new approach to similarity testing.

\begin{figure}[t]
{\small
\setlength\tabcolsep{.5pt}
\center\footnotesize
\begin{tikzpicture}[yscale=-1,node distance=0.5cm]
    \coordinate(topleft) at (0.0,0.0);
    \coordinate(clientcoord)  at (1.0,0.10);

    \draw (topleft) rectangle ++(8,3); 

    \node(client)[align=center,minimum width=3,anchor=north] at (clientcoord) {\textbf{Client}($\image$)};
    \node(clientline) [below of=client,align=center,minimum width=3] {$\embimage\xleftarrow{\$} \Embed(\image)$};

    \draw[thick,->] (3,0.9) -- node [text width=3cm,midway,above,align=center] {$\embimage$} (5,0.9);

    \node(server)[align=center,minimum width=3,anchor=north] at ($(clientcoord)+(6,0)$) {\textbf{Server}($\imageset$)};
    \node(serverline)[below of=server,align=center,minimum width=3] {$\bucketset  \gets \Sim(\embimage,\imageset)$};
    \node(protrect)[draw,rectangle,dashed,rounded corners=3,minimum width=3.5cm,minimum height=1.5cm, anchor=north west] at (2.25,1.3) {};
    \node[align=center,anchor=north] at (4,1.3) {Similarity protocol};
    \draw[thick,<-] (3,2.0) -- (5,2.0);
    \draw[thick,->] (3,2.2) -- (5,2.2);
    \draw[thick,<-] (3,2.4) -- (5,2.4);

    \node(inputw)[align=center,minimum width=1cm, anchor=north] at (1,1.5) {$\image$};
    \node(inputb)[align=center,minimum width=1cm, anchor=north] at (7,1.5) {$\bucketset$};
    \draw[thick,->] (inputw) |- (protrect);
    \draw[thick,->] (inputb) |- (protrect);

\end{tikzpicture}
\caption{\label{fig:sbb-overview} Two-stage framework for using similarity-based bucketization to determine if an image $\image$ is similar to one in an image dataset $\imageset$.
}
}
\end{figure}
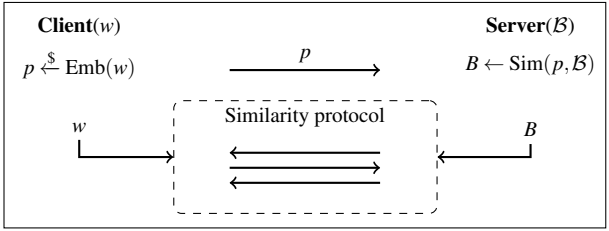

\section{Similarity-Based Bucketization}
\label{sec:sbb}

To enable efficient, privacy-preserving client-side similarity testing, 
we take inspiration from previous work that
used bucketization to support efficient private set membership testing~\cite{thomas2019protecting,li2019protocols}.
These will, however, not work for image similarity testing.
We therefore introduce a new two-step framework, as shown in \figref{fig:sbb-overview}.
It first enables scaling by utilizing what we call similarity-based
bucketization (SBB) to gather a subset $\bucketset \subseteq \imageset$, called a
bucket, of possibly relevant images. 
The second step is to perform a similarity testing
protocol over the bucket; we explore how SBB can provide scaling improvements
for several different similarity testing protocols.
In this section we introduce coarse embeddings, which allow crude
similarity comparisons, rather than the granular ones that regular similarity
hashes provide.
A summary of the notation we use throughout this paper
appears in \tabref{tab:notation}.  For simplicity, we refer to the similarity
testing server as the server.

\paragraph{Formalizing SBB.}
We formalize embeddings first. 
A similarity embedding method is a function $\simhash\Colon \imagespace\rightarrow\bits^\len$ for the space of images~$\imagespace$ and where $\len$ is a parameterizable length.
Some embeddings map images to $\R^\len$ (or a suitably granular approximation
of it), but we focus on bit strings unless explicitly mentioned
otherwise. Associated with $\simhash$ is a distance measure 
$\Delta\Colon \bits^\len \times \bits^\len \to \Z$.
We focus on
$\Delta$ being Hamming distance. 
Most often one sets an application-dependent threshold $\threshold$ as the definition of similarity,
and builds $\simhash$ so that matches with distance values smaller than $\threshold$ indicate the images
depicted are perceptually similar. 

An example $\simhash$ is the aforementioned PDQHash~\cite{pdqhash}.
PDQHash was designed to capture the ``syntactic'' similarity between images.
Syntactic similarity captures if two images are essentially the same,
e.g., the same image but of different quality, or rotated slightly.
This is different from semantic similarity,
which focuses on whether images share the same features, e.g., the same person.
Algorithms designed for syntactic similarity also include PhotoDNA~\cite{photodna} and pHash~\cite{zauner2010implementation}.
PDQHash first converts a given image $\image$ from RGB to luminance, then uses two-pass Jarosz filters to compute a weighted average of $64\times64$ subblocks of the luminance image.
Given the $64\times64$ downsample,
the algorithm computes a two-dimenstional discrete cosine transform~(DCT),
and keeps only the first $16$ slots in both X and Y directions.
After that, each entry of the $16\times16$ DCT output is transformed into a binary bit after being compared to the median, with $1$ indicating larger than the median and $0$ indicating otherwise.

\begin{table}[t]
\centering
\footnotesize
\begin{tabular}[t]{lp{0.35\textwidth}}
\toprule
Symbol & Description\\
\midrule
$\image$ / $\imagespace$ & an image / set of all images \\
$\imageset$ & set of images held by server\\
$\dist$ & distribution from which images are sampled\\
\midrule
$\len$ & length of similarity embedding \\
$\simhash$ & similarity embedding method \\
$\simimage$ / $\metricspace$ & result of similarity embedding / set of all such results\\
$\Delta$ & distance function between two similarity embeddings\\
$\threshold$ & distance threshold for similarity matching using $\Delta$\\
$\dist_{\simhash}$ & distribution of similarity embeddings induced by $\simhash$ \\
\midrule
$\Embed$ & coarse embedding algorithm\\
$\Sim$ & coarse embedding similarity algorithm\\
$\embimage$ & an output of $\Embed$\\
$\emblen$ & length of $\embimage$\\
$\bucketset$ & candidate bucket generated from $\Sim$\\
\midrule
$\flipbias$ & parameter of $\EmbedScheme_{LSH}$, flipping bias\\
$\coarsethreshold$ & coarse threshold of $\Sim$\\
\bottomrule
\end{tabular}
\caption{\label{tab:notation} Notation frequently used in this paper.}
\end{table}

\textbf{Coarse embedding schemes.}
To allow bucketization via similarity,
we define a coarse embedding scheme $\EmbedScheme = (\Embed,\Sim,(\imagespace,\distimage))$,
as a pair of algorithms and
an associated metric space. 
The (possibly randomized) embedding algorithm $\Embed(\image)$ takes as input a
value $\image \in \imagespace$ and outputs a value $\embimage \in \bits^\emblen$.
Here $\emblen$ is a configurable parameter.  
We call $\embimage$ the embedding of $\image$, or simply the embedding when
$\image$ is clear from context.  The deterministic algorithm
$\Sim(\embimage,\image')$ takes as input $\embimage \in \bits^\emblen$ and
$\image' \in \imagespace$ and outputs a bit. The bit being one 
indicates that the embedding of $\image'$ is similar to the embedding of $\image$,
which is denoted as $\embimage$.  It will be convenient to abuse notation, by letting 
$\Sim(\embimage,\imageset)$  be defined to output the 
set $\{\image' \mid \image' \in \imageset \land \Sim(\embimage,\image') = 1\}$.

One idea for a coarse embedding scheme would be to simply use $\simhash$ directly,
but with smaller 
$\len$ and smaller $\threshold$. 
To be specific,
using PDQHash as an example,
a coarse PDQHash scheme $\EmbedScheme_{cPDQ} = (\Embed_{cPDQ},\Sim_{cPDQ},(\imagespace,\distimage))$ can be implemented as follows:
$\Embed_{cPDQ}(\image)$ computes the hash of $\image$ on the first $4\times4$ slots of the DCT output, rather than $16\times16$ of the output,
producing a $16$-bit binary string.
The $16$-bit value can then provide much cruder similarity comparison.
Then $\Sim_{cPDQ}(\embimage,\imageset)$ iterates over all $\image'
\in \imageset$, hashes them,
and returns those with distance smaller than a coarse threshold~$\coarsethreshold$ as a bucket $\bucketset$.
Unfortunately this scheme doesn't meet our privacy
goals,
as we will explore in detail in \secref{sec:expr}.

\textbf{Correctness and compression efficiency.}
We define $\distimage$ via an existing similarity embedding $\embed$,
i.e., $\distimage(\image,\image')  = \Delta(\embed(\image),\embed(\image'))$.  
We say that a coarse embedding scheme is $(\threshold,\epsilon, \dist)$-correct 
if,
for an image $\image$ sampled from $\dist$, a distribution over $\imagespace$, and for any $\image'$ such that
$\distimage(\image,\image') < \threshold$, we have that
$\Prob{\Sim(\Embed(\image), \image')=1} \geq 1 - \epsilon$,
where the probability is taken over the random coins used by $\Embed$ and the choice of $\image$ from $\dist$.
A trivial coarse embedding scheme is to just use $\embed$ itself, which would
be $(\threshold,0, \dist)$-correct for any $\threshold$ and $\dist$. But as
mentioned, doing this will not provide the desired privacy. 

Another type of trivial coarse embedding scheme is to have~$\Sim$
always output one. Then $\Embed$ could output a fixed constant value regardless
of input,
meaning nothing leaks about $\image$. 
This would also be $(\threshold,0,\dist)$ correct for arbitrary $\threshold$ and any given $\dist$, but 
won't be useful because, in our SBB application, the bucket would end up being
the entire set $\imageset$.

We define a compression efficiency metric as follows.
A coarse embedding scheme is
$(\imageset,\alpha,\dist)$-compressing if for a distribution $\dist$ over $\imagespace$,
$\imageset \subseteq \imagespace$, $\image$ drawn from $\dist$ we have that
  $\Ex{\frac{|\bucketset|}{|\imageset|}} \leq \alpha$
where $\bucketset = \{\image' \mid \image' \in \imageset,  \Sim(\Embed(\image),\image')=1\}$ and the probability space is over the choice of $\image$ from $\dist$ and the coins used by $\Embed$.
This measures the average ratio of bucket size to dataset size.

\paragraph{LSH-based coarse embedding.}
\label{sec:algo-detail}
We propose a coarse embedding scheme that is based on locality sensitive hashing~(LSH)~\cite{gionis1999similarity}.
An LSH function family allows approximate nearest neighbour search with high-dimensional data.
Formally,
the scheme $\EmbedScheme_{LSH}=(\Embed_{LSH},\Sim_{LSH},(\imagespace,\distimage))$ is defined as follows~(see also \figref{fig:sbb-protocol}).
Let $\indexfuncset$ be a family of hash functions that maps points from a high-dimensional input space $\inputspace$ into a hash universe $\hashspace$ of lower dimension.
When $\inputspace=\bits^{\len}$ and $\Delta$ is Hamming distance, 
the construction of an LSH function family is intuitive.
For an $\len$-bit string $\simimage$,
we denote the individual bits as
$\simimage_1,\ldots,\simimage_\len$.
An indexing function is a map
$\indexfunc\Colon\bits^\len\to\bits$ and we let $\indexfuncset$ be the
set of all index functions,
which is the LSH function family.

In our context,
we randomize the selection of LSH functions for every individual query,
and add noise to ensure privacy.
$\Embed_{LSH}$ takes an image 
$\image$ as input, 
and computes the similarity embedding of it via $\simimage \gets \embed(\image)$.
In our implementation,
we use PDQHash for $\embed$.
Our protocol works on other types of embedding functions that use Hamming distance as a metric,
such as pHash~\cite{zauner2010implementation}.
We sample 
$\emblen$ bits from $\simimage$ by sampling $\emblen$ LSH functions without replacement
and flip each bit with probability~$\flipbias$.
The resulting embedding with added noise and the indices are shared with the server.
The server performs $\Sim_{LSH}$ by
comparing the received bits to the corresponding bits of $\embed(\image_i)$
for each $\image_i \in \imageset$, adding $\image_i$ to the bucket $\bucketset$
if sufficiently many of these bits match.

To formalize this we abuse notation slightly.
We denote $\indexfunc$ as a map $\bits^\len \to \bits^\emblen$,
a combination of $\emblen$ functions sampled uniformly from $\indexfuncset$ without replacement. 
Similarly one can easily encode an indexing function as a set of indexes; we treat $\indexfunc$ both as a function and its encoding.
We let $\flip_\flipbias$ be the randomized algorithm that takes as input
a bit string $\embimage$ and outputs $\pnoise$ of the same length, setting $\pnoise_i = \embimage_i$ with probability
$1 - \flipbias$ and $\pnoise_i = \lnot p_i$ with probability $\flipbias$. 
The full algorithms for $\EmbedScheme_{LSH}$ are shown in
\figref{fig:sbb-protocol}.
We use a threshold $\coarsethreshold$ for
the Hamming distance over the randomly selected indexes.
We formally analyze correctness of $\EmbedScheme_{LSH}$ in \apref{sec:correctness}.
Different choices of the parameter sets of $\EmbedScheme_{LSH}$, i.e., embedding length $\emblen$, flipping bias $\flipbias$, and coarse thresholds $\coarsethreshold$
result in different combination of privacy loss,
correctness, and bucket compression rate.
We explore this trade-off in \secref{sec:expr}.

\begin{figure}[t]
  \centering
  \fpage{.45}{
    \hpagess{.35}{.48}{
    \underline{$\Embed_{LSH}(\image)$}\\[1pt]
    $\simimage \gets \embed(\image)$\\
    $\indexfunc \getsr \indexfuncset$\\
    $\embimage \getsr \flip_\flipbias(\indexfunc(\simimage))$\\
    Return $(\indexfunc,\embimage)$
  }{
    \underline{$\Sim_{LSH}((\indexfunc,\embimage),\imageset)$}\\[1pt]
    $\bucketset \gets \{\}$\\
    For $\image \in \imageset$:\\
    \myInd If $\Delta(\embimage,\indexfunc(\embed(\image))) \le \coarsethreshold$ then\\
    \myInd\myInd $\bucketset \gets \bucketset \cup \{\image\}$\\
    Return $\bucketset$
  }
}
  \caption{Coarse embedding scheme $\EmbedScheme_{LSH}$.}
  \vspace{-0.3cm}
  \label{fig:sbb-protocol}
\end{figure}

One limitation of $\EmbedScheme_{LSH}$ arises should an adversary be able to
collect many queries that it knows are for the same image.  Eventually it will
see all bit locations, and even have enough samples to average out the noise
(e.g., via a majority vote for each bit location). We discuss this further in
\secref{sec:limitations}.

\paragraph{Similarity protocols.}
\label{sec:similarity_protocol}
A coarse embedding scheme will not suffice to perform a full similarity
check. Instead, we compose such a scheme to perform SBB with a
similarity protocol where the server uses the resulting bucket $\bucketset
\getsr \Sim(\Embed(\image),\imageset)$.
The composition achieves privacy levels related to the protocol's for $\imageset$,
and correctness proportional to the product
of the coarse embedding and the protocol's correctness.
We discuss some examples and their properties here. 
These examples ensure perfect correctness,
hence the correctness of the composition depends solely on that of the coarse embedding.
In \secref{sec:end-to-end}, we show that 
for various similarity protocols,
both runtime efficiency and bandwidth are largely improved when combined with SBB.

\textbf{\emph{Similarity embedding retrieval.}} A pragmatic similarity protocol has the server
send to the client the similarity embeddings of all the elements in the bucket,
i.e., send $\simimage_1,\ldots,\simimage_{|\bucketset|}$ where $\simimage_i\
= \embed(\image_i)$ for each $\image_i \in \bucketset$. The client can then
compute $\embed(\image)$ and compare against each $\simimage_i$.  
This approach reduces the
confidentiality for the server's dataset, since now clients learn all the similarity embeddings 
in $\imageset$ that fall into the bucket.  
It may also reduce
resistance to evasion attacks,
but in contexts where client privacy is paramount
this simple protocol already improves on existing approaches.

\textbf{\emph{Secure-sketch similarity protocol.}} We can improve server confidentiality via a
secure-sketch-based~\cite{dodis2004fuzzy} similarity protocol.
The protocol ensures that the client can only learn the similarity hashes that are
close to a client-known value. 
If images in the database have
sufficiently high min-entropy then the secure sketch ensures that the client
cannot learn it. This assumption may not always hold (most obviously in the case
that the client has a similar image), in which case confidentiality falls back
to that achieved by similarity embedding retrieval.
We defer details and formalization to \apref{sec:sssp}.

\textbf{\emph{2PC similarity protocols.}} Finally, one may compose SBB with an
appropriate 2PC for similarity comparisons.  Such an approach provides better
confidentiality for $\imageset$, but at the cost of larger bandwidth and
execution time.  We experiment with two frameworks: CrypTen~\cite{crypten2020}
and EMP~\cite{emp-toolkit}.  CrypTen is a secret-sharing-based semi-honest MPC framework for
Python that is geared toward machine learning applications.  CrypTen currently relies on a trusted third party
for certain operations, including generating Beaver multiplication
triples~\cite{beaver1991efficient}.  Generation of Beaver triples using
Pailler~\cite{paillier1999public} is actively under development.  EMP is a
circuit-garbling-based generic semi-honest 2PC framework that is implemented in
C++.

Both frameworks above target semi-honest security. One could also compose
SBB with a maliciously secure 2PC protocol, with caveat that a malicious server
is not bound to correctly execute the SBB $\Sim$ algorithm and so could deviate
by adding arbitrary values to the bucket. In our context, such an attack can
anyway be performed by just modifying $\imageset$ in the first place, but this
could be relevant in future work, particularly as it relates to accountability
mechanisms that monitor for changes to~$\imageset$.

\section{Privacy of Coarse Embeddings}  
\label{sec:privacy-goal}

In this section, we detail our framework for reasoning about privacy threats against 
coarse embeddings.
Our framework is 
designed to analyze the adversary's confidence in assessing a predicate being true or not, when given one or multiple client requests as input.
Here we only consider client privacy; 
privacy of the server's dataset can be achieved by
composing SBB with a suitable similarity protocol (see
\secref{sec:similarity_protocol}).

\paragraph{Proposed security measures.} 
We consider settings where an adversary receives the embedding(s) of one or more
images, and wants to infer some predicate over the images. Let $\imagespace^q$
be the Cartesian product of $q$ copies of $\imagespace$. We denote tuples of
images in bold, $\vecimage \in \imagespace^q$ and $\vecimage[i] \in
\imagespace$ for $i \in [1,q]$. 
Let $\Embed(\vecimage)$ be the result of running $\Embed$
independently on each component of $\vecimage$, denoted as $\vecembimage$. That is, $\vecembimage \getsr \Embed(\vecimage)$ is
shorthand for $\vecembimage[i] \getsr \Embed(\vecimage[i])$ for $i \in [1,q]$.

To start, consider a distribution $\dist$ over $\imagespace^q$ and a predicate
$f\Colon \imagespace^q\rightarrow\{\false,\true\}$.
We want to understand the ability of
an adversary to infer $f(\vecimage)$ when given $\Embed(\vecimage)$ for
$\vecimage$ drawn from $\imagespace^q$
according to $\dist$. As an example, let $q=1$ and have~$f$ 
indicate whether a client image has the same perceptual hash value with that of another image that is chosen by the adversary. 
We'd like to have a guarantee that revealing $\Embed(\vecimage)$ doesn't allow
inferring that the images are similar with high confidence. 
We refer to a tuple $\setting = (\dist,\imagespace^q,f)$ as a privacy setting. 

We provide three measures of adversarial
success: accuracy, precision, and area under the receiver-operator curve (AUC),
thereby adapting traditional
measures of efficacy for prediction tasks to our adversarial setting.

\textbf{Accuracy.}
Let $\advAacc$ be a randomized
algorithm, called an accuracy adversary. We define a probabilistic experiment that
tasks $\advA$ with inferring $f(\vecimage)$ given $\Embed(\vecimage)$ for
$\vecimage$ drawn according to $\dist$. This probability space is over the coins
used to sample $\vecimage$, to run $\Embed$ a total of~$q$ times, and to run
$\advAacc$. 
We let ``$\advAacc(\Embed(\vecimage)) = f(\vecimage)$'' be the event that
$\advAacc$ outputs the correct value of the predicate. 
We write this as a pseudocode game
$\PRED_{\EmbedScheme,\setting}$ shown in \figref{fig:prec-security-game}, where
the returned value captures the event that $\advA$ succeeds.  
For skewed distributions, the trivial adversary that ignores its input and simply predicts the most likely predicate value may achieve high accuracy. We therefore define the advantage of $\advAacc$ as the improvement over that trivial approach:
\begin{align*}
  \epsilonAcc &= \frac{\Prob{\advAacc(\Embed(\vecimage)) = f(\vecimage)} - \baseline}{1- \baseline}\;,
\end{align*}
where $\baseline = max(\Prob{f(\vecimage)=1}, \Prob{f(\vecimage)=0})$.

\textbf{Precision.} For adversaries that are mainly interested in inferring positive instances,
$f(\vecimage)=\true$,
accuracy may appear misleading in cases with high skew,
i.e.,  when $f(\vecimage) = \false$ happens almost always~\cite{manning1999foundations}.
In our running
example, we expect that in practice most images handled by clients will be
distinct from the adversary-chosen one.

\begin{figure}[t]
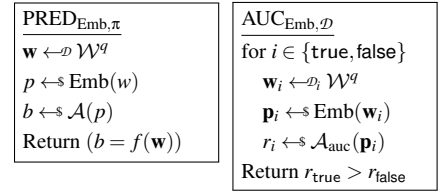

  \centering
  \hfpagess{.14}{.14}{
    \underline{$\PRED_{\Embed,\setting}$}\\[1pt]
    $\vecimage \getdist{\dist} \imagespace^q$\\
    $\embimage \getsr \Embed(\image)$\\
    $b \getsr \advA(\embimage)$\\
    Return $(b = f(\vecimage))$
    }{
    \underline{$\AUC_{\Embed,\dist}$}\\[1pt]
    for $i \in \{\true,\false\}$\\ 
    \myInd $\vecimage_i \getdist{\dist_i} \imagespace^q$\\
    \myInd $\vecembimage_i \getsr \Embed(\vecimage_i)$\\
    \myInd $r_i \getsr \advAauc(\vecembimage_i)$\\
    Return $r_{\true} > r_{\false}$
  }
  \caption{Pseudocode games for measuring embedding privacy. Here $f$ represents
  a predicate.}
  \label{fig:prec-security-game}
  \vspace{-0.3cm}
\end{figure}

We therefore also provide two other security measures. 
First, we measure the precision of a non-trivial adversary in inferring $f(\vecimage)$.
By non-trivial,
we mean that the adversary has to predict $f(\vecimage)=\true$ at least once.
We use the
same probability space as in the previous definition. 
To emphasize that the
best adversary for achieving high precision may differ from the best one for
maximizing accuracy improvement,  we use $\advAprec$ to denote the adversary
when considering precision.
We want to measure the probability that $\advAprec$ succeeds, conditioned on
$\advAprec$ outputting
$\true$. 
We denote this by $\CondProb{f(\vecimage) =
\true}{\advAprec(\Embed(\vecimage))=\true}$.
To prevent $\advAprec$ from using the trivial strategy of predicting all events as negative, 
we define an affiliate concept of recall $\epsilonRecall$ as
\bnm
  \epsilonRecall = \CondProb{\advAprec(\Embed(\vecimage))=\true} {f(\vecimage) = \true}\;.
\enm
We will restrict attention to adversaries $\advAprec$ for which $\epsilonRecall$
exceeds some threshold, e.g., $\epsilonRecall > 0\%$. 
We let
\bnm
  \epsilonPrecR{\epsilonRecall>\epsilonRecallThresh} = \CondProb{f(\vecimage) = \true}{\advAprec(\Embed(\vecimage))=\true} 
\enm
denote the precision advantage for some adversary $\advAprec$ that achieves
$\epsilonRecall > \epsilonRecallThresh$,
with the exception of $\epsilonRecallThresh=100\%$,
where the restriction is set as $\epsilonRecall=100\%$.

\textbf{AUC.}
Precision captures the adversary's confidence in predicting the positive class,
i.e., the likelihood of $f(\vecimage)$ being $\true$ when the adversary predicts it to be true.
However, it does not capture the adversary's confidence regarding predicting the negative class.
We therefore
finally formalize a notion of AUC,
where recall that AUC is the area
under the receiver-operator curve, a popular measure of classifier efficacy.
At a high level, AUC-ROC indicates the classifier's capability in differentiating positive classes from negative ones.
For a setting $\setting = (\dist,\imageset^q,f)$, let $\dist_i$ be the
distribution $\dist$ over $\imagespace^q$ conditioned on $f(\vecimage) = i$ for
$i \in \{\true,\false\}$.  Then for an adversary $\advAauc$ that outputs a real value in
$[0,1]$ we measure the probability that $\advAauc(\Embed(\vecimage_{\true})) >
\advAauc(\Embed(\vecimage_{\false}))$ where $\vecimage_i$ is drawn from $\imageset^q$
according to $\dist_i$. The probability is over the independent choices of
$\vecimage_{\true}$ and $\vecimage_{\false}$, as well as the coins used by the $2q$
executions of $\Embed$ and two executions of $\advAauc$.
We provide a pseudocode game $\AUC_{\Embed,\setting}$ describing this
probability space in \figref{fig:prec-security-game}.
Then we define the advantage of an AUC adversary $\advAauc$ by
\bnm
   \epsilonAuc =  2\cdot\Prob{\advAauc(\Embed(\vecimage_{\true})) >
   \advAauc(\Embed(\vecimage_{\false}))} - 1 \;.
\enm
This formulation uses a well-known fact~\cite{cortes2004auc,agarwal2005generalization}
about AUC that it is equal to the probability that a scoring algorithm (in our
case, the adversary) ranks positive-class instances higher than negative-class
instances. For simplicity, we ignore ties ($\advAauc$ outputting the same value
in each case). Without loss of generality, we can assume that the AUC adversary $\advAauc$ wins
the game with probability greater than or equal to 0.5, and so the normalization
maps to the range $[0,1]$. (This corresponds to the classic Gini coefficient.)

\textbf{Possible predicates.}
We focus on the matching predicate in our analyses.
An adversary chooses an image $\image_{adv}$,
and wishes to determine if the client request $\Embed(\image_{c})$ corresponds to an image that is very similar to $\image_{adv}$,
i.e., $\embed(\image_c)=\embed(\image_{adv})$.
Had the adversary the confidence to assert that there's a match,
they can recover the content from the submitted request.
Such an attack is trivial in hashing-based similarity testing services,
when the clients are required to submit similarity hashes as their requests.
That said, our framework can be used to analyse other predicates.
For example, an adversary may want to infer if $q$ different client requests all correspond to similar content.
We leave such analyses to future work.

\textbf{Discussion.}
We have omitted placing computational limits on adversaries, which would be useful in cases where embedding schemes rely on cryptographic tools --- our mechanisms do not. A computational treatment is a straightforward extension to our framework.
The security games underlying our measures are conservative, and in particular
we assume that the adversary has perfect knowledge of the distribution $\dist$ as well as $\dist$'s 
support,
which is unlikely in practice.
While we do not explicitly model side information that an adversary might have about a client's image,
it is possible to include it indirectly in this framework,
for example,
by changing the distribution or modifying the privacy predicate.

\paragraph{Bayes optimal adversaries.}
\label{sec:optimal}
To allow simulations that evaluate privacy, we 
focus on adversaries that maximize advantage. 
Recall that we assume that the adversary knows the distribution $\dist$
from which the clients are sampling images for their requests.
Upon receiving client submitted requests $\vecembimage=\Embed(\vecimage)$,
the Bayes optimal adversary computes the \textbf{exact} likelihood of the predicate being $\true$ --- $\CondProb{f(\vecimage)=\true}{\Embed(\vecimage)=\vecembimage}$,
probabilities are over the choice of $\vecimage$ being sampled from $\imagespace^q$ and coins used by the executions of $\Embed$.

The Bayes optimal adversary for the precision metric, $\advAprec$,
chooses a threshold $\probThreshold$,
such that $\advAprec(\vecimage)=\true$
if and only if $\CondProb{f(\vecimage)=\true}{\Embed(\vecimage)=\vecembimage} > \probThreshold$.
The adversary may choose $\probThreshold$ to maximize
$\epsilonPrecR{\epsilonRecall>\epsilonRecallThresh}$.
A similar strategy can be used by $\advAacc$.
However when $f(\vecimage)=\true$ is especially rare,
the adversary may achieve larger $\epsilonAcc$ by predicting all predicates using the majority class, $f(\vecimage)=\false$.
When doing so, the optimal $\epsilonAcc$ is zero.
Note that in our simulations we consider all possible threshold values for the sampled dataset, and report on the one that provides the best success rate. 
A real attacker would have to pick a threshold a priori, meaning our analyses are conservative.

The Bayes optimal adversary for $\advAauc$ doesn't have to choose a threshold~$\probThreshold$.
The adversary is given two scenarios to rank: $\vecimage_{\true}$ and $\vecimage_{\false}$,
one has $f(\vecimage_{\true})=\true$ and the other has $f(\vecimage_{\false})=\false$.
The adversary wins the game when they correctly rank the $\true$ scenario over the $\false$ one,
i.e., when $\advAauc(\vecembimage_{\true}) > \advAauc(\vecembimage_{\false})$.
As $\Embed(\vecembimage)$ is the only information that the adversary gained from our SBB protocol,
the optimal strategy to utilize the information is hence to use $\CondProb{f(\vecimage)=\true}{\Embed(\vecimage)=\vecembimage}$ as $\advAauc(\vecembimage)$.
\section{Balancing Security, Correctness, Efficiency}
\label{sec:expr}

In this section, we demonstrate how to balance security, correctness and
compression efficiency of SBB when using the LSH-based coarse embedding scheme
$\EmbedScheme_{LSH}$. 
We do so via simulations using real-world image sharing data collected from social media sites.
Using our framework, 
we evaluate the security of $\EmbedScheme_{LSH}$ with varying parameter settings.
We then fix the security requirement and explore the trade-off between
correctness and compression efficiency.

\subsection{Experimental Setup}
\label{sec:data}

\paragraph{Data collection.}
Recall that our deployment scenario in \secref{sec:deployment}, an ideal dataset
should represent the image sharing behaviors among users on an end-to-end encrypted messaging platform.
However, data of one-to-one shares among users on any private messaging platform is by definition, private.
Hence, we sought a public dataset that may act as a stand-in for our experiments.
As the dataset is publicly available, our experiment did not require review from
our IRB office.
On Twitter, users may retweet the content that they want to share with their audience.
We consider retweets on public Twitter as a proxy for user sharing in private social network.
Furthermore, the problem of misinformation, which motivated our study, is prevalent on Twitter~\cite{hindman2018disinformation}.

We were able to use a dataset from previous
work~\cite{hua2020characterizing,hua2020towards} that contains Twitter
interactions with a group of US political candidate accounts between September
13, 2018 and January~6, 2019.  The dataset includes 1,190,355 tweets with image
URLs, however, by the time of downloading images, only 485\,K images were
successfully retrieved.  We encode the retrieved images with PDQHash. There are
256,049 unique PDQHash values in total.  We collect the total number of retweets
of each tweet in November 2019: the data was available for $13\%$ of the tweets.
The total number of postings and retweets of the images we retrieved adds up to
1.2 million.
We simulate a workload for similarity testing as follows.  Any tweet and retweet
with a valid image is considered as a client request to our system.  Hence, this
experimental set has 1.2 million client requests with 256\,K unique PDQHash
values.

\begin{figure}
\begin{adjustbox}{width=0.48\textwidth}
\begin{tikzpicture}
\begin{axis}[
    xbar stacked,
	bar width=9pt,
	y=13pt,
	ymin=-0.2,
	ymax=3.2,
	xmin=0,
	xmax=1,
    enlarge y limits=0.15,
    legend style={at={(1.135,1)},
      anchor=north,legend columns=1, name=legend,font=\footnotesize},
    xlabel={Percentage of Similarity Embeddings},
    xlabel style={font=\footnotesize},
    ylabel style={font=\footnotesize},
    yticklabel style={font=\footnotesize},
    xticklabel style={font=\footnotesize},
    ytick={0, 1, 2, 3},
    yticklabels={{$\threshold=0$}, {$\threshold=32$},
    {$\threshold=64$}, {$\threshold=70$}},
    xticklabel={\pgfmathparse{\tick*100}\pgfmathprintnumber{\pgfmathresult}\%},
    ]

\addplot[fill=clr2_1!25, draw=white, xbar] plot coordinates {(0.0996,0) (0.0835,1) 
  (0.0805,2) (0.08,3)};
\addlegendentry{$1$}
\addplot[fill=clr2_1!50, draw=white, xbar] plot coordinates {(0.8401,0) (0.7480,1) 
  (0.7008,2) (0.6901,3)};
\addlegendentry{$(1, 10]$}
  
 \addplot[fill=clr2_1!75,draw=white, xbar] plot coordinates {(0.0577,0) (0.14,1)
  (0.1707,2) (0.1777,3)};
\addlegendentry{$(10, 100]$}
 \addplot[fill=clr2_1!100, draw=white, xbar] plot coordinates {(0.0026,0) (0.0285,1) 
  (0.048,2) (0.0523,3)};
\addlegendentry{$> 100$}

\end{axis}
\end{tikzpicture}
\end{adjustbox}
\caption{\label{fig:neighborhood} $\threshold$-Neighborhood size distribution for different $\threshold$.}
\vspace{-0.4cm}
\end{figure}
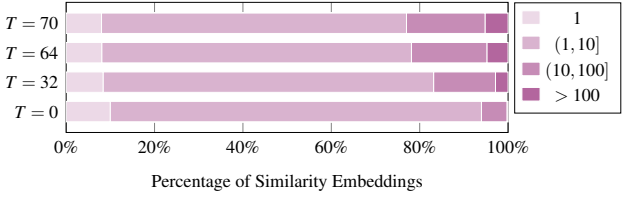

\textbf{Dataset statistics.}
Users on social media share similar images frequently.
The $\threshold$-neighborhood size of an image~$\image$ is the number of images that  are $\threshold$-similar to it.
Two images are $\threshold$-similar if and only if their similarity embeddings have a Hamming distance smaller than $\threshold$.
We choose the values of $\threshold$ according to the recommendations from the
white paper on PDQHash~\cite{pdqhash},
where 32 and 70 were specified as the lower and upper bounds of recommended similarity thresholds.
We also include $\threshold=0$ and $\threshold=64$ for comparison.

Figure~\ref{fig:neighborhood} shows the distribution of
$\threshold$-neighborhood sizes~(shades of color) of the images in client requests,
with different~$\threshold$~(in different rows).
The lightest shades~(left) are requests with neighborhood size of one, i.e., the neighborhood only contains the single image.
The following darker shades are the images with neighborhood size in the range $(1, 10]$, $(1,100]$, and $(100,\infty)$.
The bottom row shows the neighborhood size distribution with $\threshold=0$.
In our dataset,
most of the images in the client requests~($84\%$) share the same similarity embedding with more than one, but fewer than $10$ neighbors.
The distributions of neighborhood size are mostly similar to each other,
especially for $\threshold=64$ and $\threshold=70$.
Naturally, with a larger threshold, there are more requests containing images with a larger neighborhood size. 
For example, only $2.85\%$ of all request images have a neighborhood size larger than $100$ with $\threshold=32$~(third row from top, darkest column with purple),
while $5.23\%$ of the requests satisfy the same condition when $\threshold=70$~(first row from top, darkest column with purple).

\paragraph{Implementation.}
We compare the privacy of SBB
when using $\EmbedScheme_{LSH}$ with different embedding lengths~$\emblen$
to the baseline method $\EmbedScheme_{cPDQ}$.
We focus on the security guarantees against the matching attack and explain our implementation details.

\begin{figure*}
\centering
\begin{adjustbox}{width=\textwidth}
   \begin{tikzpicture}
     \begin{axis}[
            name=plot1,
    	    legend columns = 1,
    	    ticklabel style = {font=\footnotesize},
    	    ylabel={Adversarial advantage},
            ylabel style={font=\footnotesize, at={(axis description cs:0.08,.5)},anchor=south},
    	    xlabel style={font=\footnotesize},
    	    xlabel={Coarse embedding length $\emblen$},
    	    xtick = {8,9,10,11,12,13,14,15,16,17.5},
            xticklabels = {$8$,$9$,$10$,$11$,$12$,$13$,$14$,$15$,$16$},
            xticklabel style={align=center},
            yticklabel={\pgfmathparse{\tick*100}\pgfmathprintnumber{\pgfmathresult}\%},
    		height=4cm,
    		width=0.45\textwidth,
    		scaled y ticks=false,
    		legend style={/tikz/every even column/.append style={column sep=0.5cm}, at={(1.01,0.48)},anchor=south west,
    		font=\footnotesize},
    	]
        \addplot[clr3_1, only marks,error bars/.cd, y dir=both, y explicit] coordinates {
(8.0000, 0.9981) += (0, 0.0002) -= (0, 0.0002)
(9.0000, 0.9991) += (0, 0.0001) -= (0, 0.0001)
(10.0000, 0.9995) += (0, 0.0001) -= (0, 0.0001)
(11.0000, 0.9998) += (0, 0.0001) -= (0, 0.0001)
(12.0000, 0.9999) += (0, 0.0000) -= (0, 0.0000)
(13.0000, 0.9999) += (0, 0.0000) -= (0, 0.0000)
(14.0000, 1.0000) += (0, 0.0000) -= (0, 0.0000)
(15.0000, 1.0000) += (0, 0.0000) -= (0, 0.0000)
(16.0000, 1.0000) += (0, 0.0000) -= (0, 0.0000)
    	};  		
      \addplot[clr3_2, only marks,error bars/.cd, y dir=both, y explicit,] coordinates {
(8.0000, 0.3426) += (0, 0.0192) -= (0, 0.0192)
(9.0000, 0.5306) += (0, 0.0232) -= (0, 0.0232)
(10.0000, 0.6651) += (0, 0.0238) -= (0, 0.0238)
(11.0000, 0.8103) += (0, 0.0333) -= (0, 0.0333)
(12.0000, 0.8864) += (0, 0.0192) -= (0, 0.0192)
(13.0000, 0.9443) += (0, 0.0241) -= (0, 0.0241)
(14.0000, 0.9706) += (0, 0.0133) -= (0, 0.0133)
(15.0000, 0.9907) += (0, 0.0050) -= (0, 0.0050)
(16.0000, 0.9930) += (0, 0.0046) -= (0, 0.0046)
       };

    	\addplot[clr3_3, only marks,error bars/.cd, y dir=both, y explicit,] coordinates {
(8.0000, 0.0000) += (0, 0.0000) -= (0, 0.0000)
(9.0000, 0.1186) += (0, 0.0715) -= (0, 0.0715)
(10.0000, 0.4931) += (0, 0.0544) -= (0, 0.0544)
(11.0000, 0.7620) += (0, 0.0544) -= (0, 0.0544)
(12.0000, 0.8709) += (0, 0.0247) -= (0, 0.0247)
(13.0000, 0.9397) += (0, 0.0291) -= (0, 0.0291)
(14.0000, 0.9694) += (0, 0.0145) -= (0, 0.0145)
(15.0000, 0.9905) += (0, 0.0051) -= (0, 0.0051)
(16.0000, 0.9929) += (0, 0.0047) -= (0, 0.0047)
    	};
    	\addplot[clr3_1, only marks, mark=diamond*,mark options={scale=2}] coordinates {(17.5, 0.964756)};
    	\addplot[clr3_2, only marks, mark=diamond*,mark options={scale=2}] coordinates {(17.5, 0.965956)};
    	\addplot[clr3_3, only marks,mark=diamond*,mark options={scale=2}] coordinates {(17.5, 0.999980)};
    	\addplot[red, line width=1pt] coordinates {
    	(17, 0.85) (18, 0.85)};
    	\addplot[red, line width=1pt] coordinates {
    	(17, 0.85) (17, 1.15)};
    	\addplot[red, line width=1pt] coordinates {
    	(18, 0.85) (18, 1.15)};
    	\addplot[red, line width=1pt] coordinates {
    	(17, 1.15) (18, 1.15)};
    	\legend{$\epsilonAuc$, $\epsilonPrecR{\epsilonRecall=100\%}$, $\epsilonAcc$}
        \node[
            anchor=north west,
            align=left,
            font=\footnotesize
        ] at (axis cs:16.5,.85)
        {$\EmbedScheme_{cPDQ}$\\$\emblen=16$};
    \end{axis}

         \begin{axis}[
            name=plot2,
            at=(plot1.right of south east), anchor=left of south west,
    	    legend columns = 1,
    	    xshift=0.3cm,
    	    ticklabel style = {font=\footnotesize},
    	    legend style={nodes={scale=.95, font=\footnotesize}}, 
    	    xmin=-0.05,
    	    xmax=0.4,
	    ymin=-0.1,
	    ymax=0.75,
            xticklabel style={align=center},
            yticklabel={\pgfmathparse{\tick*100}\pgfmathprintnumber{\pgfmathresult}\%},
            ylabel style={font=\footnotesize, at={(axis description cs:0.1,.5)},anchor=south},
            xlabel style={font=\footnotesize},
            ylabel={$\epsilonPrecR{\epsilonRecall > \epsilonRecallThresh}$},
            xlabel={Flipping bias $\flipbias$},
    		height=4cm,
    		width=0.45\textwidth,
    		scaled y ticks=false,
    		legend style={/tikz/every even column/.append style={column sep=0.25cm}, at={(1.18,1)},anchor=north, font=\footnotesize}
    	]
   		
      \addplot[clr5_1,mark=*,line width=1pt,error bars/.cd, y dir=both, y explicit,] coordinates {
(0.0000, 0.5306) += (0, 0.0232) -= (0, 0.0232)
(0.0500, 0.4049) += (0, 0.0271) -= (0, 0.0271)
(0.1000, 0.2846) += (0, 0.0141) -= (0, 0.0141)
(0.1500, 0.1939) += (0, 0.0060) -= (0, 0.0060)
(0.2000, 0.1227) += (0, 0.0052) -= (0, 0.0052)
(0.2500, 0.0718) += (0, 0.0052) -= (0, 0.0052)
(0.3000, 0.0387) += (0, 0.0031) -= (0, 0.0031)
(0.3500, 0.0220) += (0, 0.0014) -= (0, 0.0014)

       };
       \addplot[clr5_2,mark=*,line width=1pt,error bars/.cd, y dir=both, y explicit,] coordinates {
(0.0000, 0.5306) += (0, 0.0232) -= (0, 0.0232)
(0.0500, 0.4049) += (0, 0.0271) -= (0, 0.0271)
(0.1000, 0.2846) += (0, 0.0141) -= (0, 0.0141)
(0.1500, 0.1270) += (0, 0.0043) -= (0, 0.0043)
(0.2000, 0.0531) += (0, 0.0016) -= (0, 0.0016)
(0.2500, 0.0320) += (0, 0.0017) -= (0, 0.0017)
(0.3000, 0.0139) += (0, 0.0007) -= (0, 0.0007)
(0.3500, 0.0083) += (0, 0.0002) -= (0, 0.0002)

       };
       \addplot[clr5_3,mark=*,line width=1pt,error bars/.cd, y dir=both, y explicit,] coordinates {
(0.0000, 0.5306) += (0, 0.0232) -= (0, 0.0232)
(0.0500, 0.4049) += (0, 0.0271) -= (0, 0.0271)
(0.1000, 0.1218) += (0, 0.0054) -= (0, 0.0054)
(0.1500, 0.0630) += (0, 0.0018) -= (0, 0.0018)
(0.2000, 0.0282) += (0, 0.0018) -= (0, 0.0018)
(0.2500, 0.0152) += (0, 0.0002) -= (0, 0.0002)
(0.3000, 0.0086) += (0, 0.0004) -= (0, 0.0004)
(0.3500, 0.0053) += (0, 0.0001) -= (0, 0.0001)
      };
  		
      \addplot[clr5_4,mark=*,line width=1pt,error bars/.cd, y dir=both, y explicit,] coordinates {
(0.0000, 0.5306) += (0, 0.0232) -= (0, 0.0232)
(0.0500, 0.1438) += (0, 0.0091) -= (0, 0.0091)
(0.1000, 0.0736) += (0, 0.0028) -= (0, 0.0028)
(0.1500, 0.0251) += (0, 0.0010) -= (0, 0.0010)
(0.2000, 0.0137) += (0, 0.0007) -= (0, 0.0007)
(0.2500, 0.0077) += (0, 0.0002) -= (0, 0.0002)
(0.3000, 0.0053) += (0, 0.0002) -= (0, 0.0002)
(0.3500, 0.0036) += (0, 0.0001) -= (0, 0.0001)
       };
       \addplot[clr5_5,mark=*,line width=1pt,error bars/.cd, y dir=both, y explicit,] coordinates {
(0.0000, 0.5306) += (0, 0.0232) -= (0, 0.0232)
(0.0500, 0.0067) += (0, 0.0010) -= (0, 0.0010)
(0.1000, 0.0031) += (0, 0.0003) -= (0, 0.0003)
(0.1500, 0.0023) += (0, 0.0002) -= (0, 0.0002)
(0.2000, 0.0021) += (0, 0.0001) -= (0, 0.0001)
(0.2500, 0.0021) += (0, 0.0000) -= (0, 0.0000)
(0.3000, 0.0020) += (0, 0.0000) -= (0, 0.0000)
(0.3500, 0.0018) += (0, 0.0005) -= (0, 0.0005)
       };
       \addplot [black, no markers, line width=1pt,dashed] coordinates {(-0.5,0.5) (0.6,0.5)};
       \node[
            anchor=north west,
            align=left,
            font=\footnotesize
        ] at (axis cs:0.27,0.67)
        {$\epsilonPrecR{\epsilonRecall > \epsilonRecallThresh}=50\%$};
       \legend{{$\epsilonRecall > 0\%$}, {$\epsilonRecall > 25\%$},
       {$\epsilonRecall > 50\%$},{$\epsilonRecall > 75\%$}, {$\epsilonRecall =100\%$}}

     \end{axis}

    \end{tikzpicture}
    \end{adjustbox}
        \caption{\label{fig:privacy} \textbf{Left}: Simulation results of the three security metrics, evaluated on (1) $\EmbedScheme_{LSH}$ with embedding length $\emblen$ from $8$ to $16$ (round markers), $\flipbias=0$ and (2) $\EmbedScheme_{cPDQ}$ (diamond markers in red box).
    \textbf{Right}: The conditioned precision metric $\epsilonPrecR{\epsilonRecall>\epsilonRecallThresh}$ of matching attack, at different recall threshold, with $\emblen=9$ and varying $\flipbias$. 
    Error bars in both plots represent the $95\%$ confidence interval.}
    
\end{figure*}

We formally define the matching attack as
$\setting_{match}^{\image_{adv}} = (\dist,\imagespace,f_{match}^{\image_{adv}})$,
where $\dist$ is the distribution over $\imagespace$
that the client requests are sampled from
and $\image_{adv}$ is an image chosen by the adversary.
For any client submitted request with an image $\image$,
the adversary wishes to learn the value of
$f_{match}^{\image_{adv}}(\image)$.
We have $f_{match}^{\image_{adv}}(\image)=\true$ if and only if the corresponding PDQHash of the two images are the same,
i.e., $\simhash(\image)=\simhash(\image_{adv})$. 
When trying to match the client image to $\image_{adv}$,
an adversary who receives the client request through a bucketization protocol is able to filter out images that are not in the bucket.
When the client image is in fact similar to $\image_{adv}$,
it should most likely be included in the same bucket by definition of correctness.
In this case,
a $\image_{adv}$ whose similarity hash value is shared more frequently than any other image in the same bucket may boost the adversary's confidence in asserting that the client image is a similar match.
Hence, 
having an image with a more popular hash value as $\image_{adv}$ increases the adversarial advantage.
In the following experiments, we use the image that has the most popular $\simhash(\image)$ as $\image_{adv}$.
The same similarity hash appeared in $0.2\%$ of all requests.

To evaluate the privacy guarantees provided by $\EmbedScheme_{LSH}$ and $\EmbedScheme_{cPDQ}$,
we iterate over all requests in our dataset and simulate the client, server, and adversary behavior,
to compute the security metrics $\epsilonAcc$, $\epsilonAuc$, and $\epsilonPrecR{\condition}$.

\textbf{Coarse PDQHash embedding scheme~($\EmbedScheme_{cPDQ}$).}
Recall the algorithm of $\EmbedScheme_{cPDQ}$ described in \secref{sec:sbb}.
The embedding algorithm $\Embed_{cPDQ}$ takes an image~$\image$ as an input,
and follows the PDQHash algorithm but with modified parameter settings,
to generate a $16$-bit coarse PDQHash.
This allows coarse grained similarity comparison.
When receiving a request $\embimage=\Embed(\image)$,
the Bayes optimal adversary follows the strategy as described in \secref{sec:optimal}.
To be specific,
the adversary computes the likelihood of the predicate being true,
$\CondProb{\simhash(\image)=\simhash(\image_{adv})}{\Embed_{cPDQ}(\image)=\embimage}$,
and makes a binary prediction based on the computed likelihood.

\textbf{LSH-based embedding scheme~($\EmbedScheme_{LSH}$).}
Recall that in \figref{fig:sbb-protocol}, the LSH-based embedding scheme $\EmbedScheme_{LSH}$ consists of two algorithms $\Embed_{LSH}$ and 
$\Sim_{LSH}$.
$\Embed_{LSH}$ takes an image~$\image$ as an input,
and outputs the selected indexing function and the resulting coarse embedding.
To be specific,
$\Embed_{LSH}(\image)$ randomly selects a length $\emblen$ indexing function $\indexfunc$ from $\indexfuncset$,
and computes $\pnoise=\flip_{\flipbias}(\indexfunc(\embed(\image)))$.
The function is parameterized by two parameters:
the length of the indexing function $\emblen$ and the bias~$\flipbias$ to flip an index that was chosen.
Note that $\emblen$ is also the length of the coarse embedding.
$\Sim_{LSH}$ takes the output from $\Embed_{LSH}$ and a dataset $\imageset$ as input,
and outputs a candidate bucket as an output.
The function has one parameter $\coarsethreshold$, a coarse threshold to choose items for the candidate bucket.

We analyse the security guarantee of $\EmbedScheme_{LSH}$ against  the Bayes optimal adversary under the setting of a matching attack $\setting_{match}^{\image_{adv}}$.
When receiving a request $\embimage$,
as in $\embimage=\Embed(\image)$,
the Bayes optimal adversary wants to predict $f_{match}^{\image_{adv}}(\image)=\true$.
The adversary bases their prediction on 
$\CondProb{\simhash(\image)=\simhash(\image_{adv})}{\Embed_{LSH}(\image)=\embimage}$,
and computes the probability by using all the information that is revealed to them: 
the indexing function~$\indexfunc$, the resulting coarse embedding~$\pnoise$ and the added noise~$\flipbias$.  We define the distribution $\distf$ of similarity hashes of images sampled from $\dist$ as follows: $\distf(x) = \Pr[\simhash(\image) = x] = \sum\limits_{\substack{\image' \in \imagespace, \simhash(\image') = x}} \dist(\image')$.

\begin{restatable}{theorem}{bayesOpt}\label{thm:bayes-opt}
Let $\Delta_I(\simimage, p) = \Delta(I(\simimage), p)$ for a similarity hash
  $\simimage \in \bits^\len$.  Consider a fixed image $\image_{adv} \in
  \imagespace$ and a sampled image $\image \getsr \dist$. Let $\simimage_{adv} =
  \simhash(\image_{adv})$ and $\simimage = \simhash(\image)$.
Then 
\begin{align*}
  &\Pr[\simimage = \simimage_{adv} \mid \Embed(\image) = (I, p)] \\
  &\hspace*{3em}=\hspace*{1em} \dfrac{\gamma^{\Delta_I(\simimage_{adv}, p)} \cdot (1 - \gamma)^{d - \Delta_I(\simimage_{adv}, p)} \cdot \distf(\simimage_{adv})}
{\sum\limits_{\simimage' \in \bits^\len} \gamma^{\Delta_I(\simimage', p)} \cdot
  (1 - \gamma)^{d - \Delta_I(\simimage', p)}\cdot \distf(\simimage')} \;,
\end{align*}
\noindent where the probability is over the coins used by $\Embed$ and the choice of $\image$ sampled from $\dist$.
\end{restatable}
We prove this theorem in \apref{ap:bayes}.
When analysing the security guarantee of $\EmbedScheme_{LSH}$,
we vary the parameter settings of $\emblen$ and $\flipbias$ only,
as~$\coarsethreshold$ has no impact on the adversarial advantage.
For any choice of~$\emblen$ and~$\flipbias$,
we randomly choose an index function $\indexfunc$,
and then execute the protocol for all requests in the dataset.
We fix $\indexfunc$ in the simulation because this information is revealed to the adversary and the adversary computes the likelihood conditioned on the indexing function.
We first repeat this process for at least $10$ times.
For parameter settings that result in large fluctuation in the results,
we repeat for another $10$ iterations.
We were able to obtain a stable estimate for all experiments after at most $20$ iterations.

\subsection{Privacy of SBB}
\label{sec:expr-attack}

\paragraph{Using different security metrics.}
To obtain a broad understanding of the information leakage from our protocol,
we present all three security metrics $\epsilonAuc$,
$\epsilonPrecR{\epsilonRecall>\epsilonRecallThresh}$, and $\epsilonAcc$ in Figure~\ref{fig:privacy}~(left).
To be specific, we show the results of using (1) $\EmbedScheme_{LSH}$ without added noise (i.e., $\flipbias=0$), but with varying embedding length $\emblen$,
and (2) $\EmbedScheme_{cPDQ}$, a baseline method that embeds a client request into a $16$-bit coarse PDQHash.
The Y axis denotes the value of the security metrics, ranging from $0$ to $100\%$.
The X axis denotes different methods used. 
From left to right,
we list $\EmbedScheme_{LSH}$ with $\emblen$ ranging from $8$ to $16$ (with round markers).
In the red box, to the very right,
the diamond markers represent the security metrics for $\EmbedScheme_{cPDQ}$.
A naive baseline of using the plaintext similarity embedding~(as mentioned in \secref{sec:overview}) achieves $100\%$ for all security metrics~(not included in figure).

In our setting,
accuracy measures the adversary's performance in predicting the correct class;
precision measures the adversary's confidence in predicting the positive class;
AUC measures the adversary's ability in differentiating negative classes from positive ones.
Our dataset is highly skewed, with more negative predicates than positive ones:
~only $0.2\%$ of all requests trigger positive matches.
This property may lead to a biased view when measured by certain metrics.

\textbf{\emph{Accuracy.}}
In datasets with a skewed distribution,
a trivial algorithm that always outputs the majority class,
i.e., $f_{match}^{\image_{adv}}(\image)=\false$,
may achieve higher accuracy than any meaningful algorithm that tries to differentiate positive cases from negative cases.
In fact, in some cases,
when experimenting with $\EmbedScheme_{LSH}$ with $\emblen=8$, we have
$\epsilonAcc=0$~(\figref{fig:privacy} on the left, green markers).
This indicates that the adversarial advantage as measured by $\epsilonAcc$ was based on the performance of the trivial algorithm,
hence was considered as none.
Meanwhile,
other metrics~($\epsilonAuc$ and $\epsilonPrecR{\epsilonRecall=100\%}$) show that the adversarial advantage is non-zero,
e.g., $\epsilonPrecR{\epsilonRecall=100\%}=37\%$ for $\emblen=8$.
Both $\EmbedScheme_{cPDQ}$ and $\EmbedScheme_{LSH}$  allow the adversary to have perfect accuracy improvement when performing the matching attack when of similar coarse embedding length.
However, applying $\EmbedScheme_{LSH}$ with a smaller coarse embedding length, 
e.g., $\emblen=10$ decreases the accuracy improvement to $45\%$.

\textbf{\emph{Precision.}}
When there's no noise added in $\EmbedScheme_{LSH}$~($\flipbias=0$),
$\epsilonPrecR{\epsilonRecall>\epsilonRecallThresh}$ remains the same regardless of varying recall threshold $\epsilonRecallThresh$.
We will expand on this point later.
With the same value of precision, a larger recall
specifies more true positive predicates that the adversary may correctly classify with high confidence,
hence larger privacy damage.
When using $\EmbedScheme_{LSH}$,
decreasing the length of coarse embeddings~($\emblen < 12$)
decreases adversarial precision,
improving security.
However, with embeddings of similar length,
$\EmbedScheme_{LSH}$ and $\EmbedScheme_{cPDQ}$ both behave poorly~($\emblen=15,16$ compared with $\EmbedScheme_{cPDQ}$).

\textbf{\emph{AUC.}}
Regardless of the embedding schemes, $\epsilonAuc$ is almost $100\%$.
The reason is that there are disproportionally many images for which the predicate evaluates as negative that can be easily differentiated from positive ones.
Hence,
most images with different PDQHashes from $\image_{adv}$ are assigned to different buckets than $\image_{adv}$.
Therefore, when measured by the adversary confidence of differentiating negative cases from positive ones,
as most of the negative cases can be distinguished correctly,
the embedding schemes behave poorly.
Note that the definition of AUC is in direct conflict with the utility of an SBB scheme,
which allows efficiently ruling out images that are not similar with a given client input.

None of the metrics ensures a lower bound for another.  While testing on all
metrics of adversarial advantage offers a better understanding, it is more
practical to focus on a specific security metric that suits the application
context.  We focus on increasing the adversary's uncertainty in classifying a
positive match. This aligns with previous work in the machine learning community
that recommends the precision-recall metric over both AUC and accuracy, when
evaluating prediction algorithms on highly imbalanced datasets, especially when
correctly predicting the positive class is valued
more~\cite{saito2015precision,raghavan1989critical,davis2006relationship,manning1999foundations}.
Hence, $\epsilonPrecR{\epsilonRecall > \epsilonRecallThresh}$ fits our purpose the best, and we focus on it
in the following discussion.

\paragraph{Adding noise.}We now evaluate the security of $\EmbedScheme_{LSH}$ with added noise,
i.e., have $\flipbias > 0$.
We fix the value of $\emblen$ with $\emblen=9$ and present results for $\epsilonPrecR{\epsilonRecall>\epsilonRecallThresh}$ in \figref{fig:privacy}~(right)
with different values of $\epsilonRecallThresh$.
The results are similar for other values of $\emblen$,
though the exact value of $\epsilonPrecR{\epsilonRecall>\epsilonRecallThresh}$ differs.
The Y-axis denotes $\epsilonPrecR{\epsilonRecall>\epsilonRecallThresh}$ and the X-axis denotes increasing value of $\flipbias$.

Precision metrics conditioned with different recall thresholds have different meanings.
For example, 
there are 2,533 requests in total that contain images similar to $\image_{adv}$.
The metric $\epsilonPrecR{\epsilonRecall>0\%}$ measures the confidence of the adversary catching one true positive case out of all, 
while $\epsilonPrecR{\epsilonRecall=100\%}$ measures the confidence of the
adversary classifying all 2,533 cases correctly.
Naturally we have   
$\epsilonPrecR{\epsilonRecall>0\%} \geq \epsilonPrecR{\epsilonRecall=100\%}$.
On the other hand,
$\epsilonPrecR{\epsilonRecall=100\%}$ indicates more advantage for the adversary when having the same value as $\epsilonPrecR{\epsilonRecall>0\%}$.
When no noise is added~($\flipbias=0$),
the adversary computes the same likelihood for all true positive cases when following the strategy as described in \secref{sec:optimal}.
Hence $\epsilonPrecR{\epsilonRecall>\epsilonRecallThresh}$ remains the same regardless of the recall threshold $r$.
With increasing $\flipbias>0$, 
it gets harder for the adversary to have high precision while maintaining large recall~(green, purple and orange lines).

Even a small value for~$\flipbias$ improves the privacy of the client request.
For example,
when $\flipbias = 0.05$ and $\emblen=9$, the likelihood that at least one value in a client request gets changed is $1-(1-\flipbias)^\emblen=37\%$.
That means that in the majority of queries, none of the client request bits will
be flipped,
nevertheless the possibility that they could have been affects adversarial
precision.
For example, this results in a drop from $53\%$~(leftmost orange node in
\figref{fig:privacy}, right chart) to $40\%$~(second green node from left) for $\epsilonPrecR{\epsilonRecall > \epsilonRecallThresh}$ for $\epsilonRecallThresh \leq 50\%$.
There's a more drastic drop when $\epsilonRecallThresh$ is larger.
When the adversary has to identify more than $50\%$ of the true positives, 
they have to lower the threshold $\probThreshold$ in predicting a positive answer,
this leads to a larger likelihood of being impacted from the uncertainty introduced by the flipping bias.
Nevertheless, when having $\flipbias \geq 0.05$, $\epsilonPrecR{\epsilonRecall>\epsilonRecallThresh}$ is smaller than $50\%$~(noted by the horizontal line) for all recall thresholds.
This indicates that for any given query that the $\advAprec$ predicts as $\true$, the adversarial success rate is lower than randomly flipping a coin.

In summary, these results show that using coarse
PDQHash $\EmbedScheme_{cPDQ}$ fails to provide privacy for clients.  The naive
solution of revealing the plaintext similarity embeddings to the server also
provides no privacy.  Different security metrics demonstrate different aspects
of adversarial advantage.  We focus on $\epsilonPrecR{\epsilonRecall>\epsilonRecallThresh}$ as its
definition fits our privacy goal the best.  For our purpose, we consider
$\epsilonPrecR{\epsilonRecall > 0\%} < 50\%$ as our security goal, i.e., when the majority of an
adversary's positive guesses~(given that there's at least one) are wrong.  Given
our empirical analysis, we suggest that a reasonable choice of parameters is embedding length $\emblen=9$ and
flipping bias $\flipbias=0.05$, but caution that the privacy performance may
vary in practice should the image distribution be very different (see
\secref{sec:limitations}).

\subsection{Correctness and compression efficiency}
\label{sec:correctness-expr}

We show the tradeoff between correctness and bucket compression rate under
the security parameters suggested above ($\emblen=9$,
$\flipbias=0.05$).
We vary the value of the coarse threshold $\coarsethreshold$,
which specifies the bucket being sent from the server side when performing $\Sim_{LSH}$.
Formally, we refer to the notion of correctness as $\epsilon$ in the definition of $(\threshold, \epsilon, \dist)$-correct,
bucket compression rate as $\alpha$ in that of $(\imageset,\alpha,\dist)$-compressing~(see \secref{sec:sbb}).
We use the dataset as~$\dist$,
which we refer to as $\dist_{twitter}$ and all the possible values of PDQHash in the dataset as~$\imageset$.
We randomly select 2 million requests with replacement and perform the protocol, 
and take the average value of the correctness and compression rate from all iterations. 

\begin{figure}
\centering
\begin{adjustbox}{width=0.45\textwidth}
   \begin{tikzpicture}
     \begin{axis}[
            name=plot1,
    	    legend columns = 1,
    	    ticklabel style = {font=\footnotesize},
    	    xmin=-0.05,
    	    xmax=1.05,
    	    ymin=0.83,
    	    ymax=1.02,
            yticklabel={\pgfmathparse{\tick*100}\pgfmathprintnumber{\pgfmathresult}\%},
            xticklabel={\pgfmathparse{\tick*100}\pgfmathprintnumber{\pgfmathresult}\%},
            ylabel = {Correctness $\epsilon$},
            xlabel = {Compression efficiency $\alpha$},
            ylabel style={font=\footnotesize, at={(axis description cs:0.08,.5)},anchor=south},
            xlabel style = {font=\footnotesize},
    		height=3.8cm,
    		width=0.45\textwidth,
    		scaled y ticks=false,
    		legend style={/tikz/every even column/.append style={column sep=0.38cm}, at={(1.02,0.25)},anchor=south west, font=\footnotesize}
    	]
        \addplot[clr5_1, only marks,error bars/.cd, y dir=both, x dir=both, y explicit, x explicit] coordinates {
(0.0210, 0.9291) += (0.0000, 0.0004) -= (0.0000, 0.0004)
(0.0932, 0.9917) += (0.0000, 0.0001) -= (0.0000, 0.0001)
(0.2573, 0.9994) += (0.0000, 0.0000) -= (0.0000, 0.0000)
(0.5002, 1.0000) += (0.0000, 0.0000) -= (0.0000, 0.0000)
(0.7430, 1.0000) += (0.0000, 0.0000) -= (0.0000, 0.0000)
(0.9070, 1.0000) += (0.0000, 0.0000) -= (0.0000, 0.0000)
(0.9791, 1.0000) += (0.0000, 0.0000) -= (0.0000, 0.0000)
    	};  		
      \addplot[clr5_2, only marks,error bars/.cd, y dir=both, y explicit, x dir=both, x explicit] coordinates {
(0.0210, 0.9052) += (0.0000, 0.0004) -= (0.0000, 0.0004)
(0.0932, 0.9837) += (0.0000, 0.0002) -= (0.0000, 0.0002)
(0.2573, 0.9977) += (0.0000, 0.0001) -= (0.0000, 0.0001)
(0.5002, 0.9997) += (0.0000, 0.0000) -= (0.0000, 0.0000)
(0.7430, 1.0000) += (0.0000, 0.0000) -= (0.0000, 0.0000)
(0.9070, 1.0000) += (0.0000, 0.0000) -= (0.0000, 0.0000)
(0.9791, 1.0000) += (0.0000, 0.0000) -= (0.0000, 0.0000)
       };

    	\addplot[clr5_3, only marks,error bars/.cd, y dir=both, y explicit,] coordinates {
(0.0022, 0.5563) += (0.0000, 0.0006) -= (0.0000, 0.0006)
(0.0210, 0.8622) += (0.0000, 0.0004) -= (0.0000, 0.0004)
(0.0932, 0.9586) += (0.0000, 0.0002) -= (0.0000, 0.0002)
(0.2573, 0.9878) += (0.0000, 0.0001) -= (0.0000, 0.0001)
(0.5002, 0.9970) += (0.0000, 0.0000) -= (0.0000, 0.0000)
(0.7430, 0.9994) += (0.0000, 0.0000) -= (0.0000, 0.0000)
(0.9070, 0.9999) += (0.0000, 0.0000) -= (0.0000, 0.0000)
(0.9791, 1.0000) += (0.0000, 0.0000) -= (0.0000, 0.0000)
    	};
    	    	\addplot[clr5_4, only marks,error bars/.cd, y dir=both, y explicit,] coordinates {
(0.0022, 0.5475) += (0.0000, 0.0006) -= (0.0000, 0.0006)
(0.0210, 0.8509) += (0.0000, 0.0004) -= (0.0000, 0.0004)
(0.0932, 0.9502) += (0.0000, 0.0002) -= (0.0000, 0.0002)
(0.2573, 0.9836) += (0.0000, 0.0001) -= (0.0000, 0.0001)
(0.5002, 0.9955) += (0.0000, 0.0001) -= (0.0000, 0.0001)
(0.7430, 0.9991) += (0.0000, 0.0000) -= (0.0000, 0.0000)
(0.9070, 0.9999) += (0.0000, 0.0000) -= (0.0000, 0.0000)
(0.9791, 1.0000) += (0.0000, 0.0000) -= (0.0000, 0.0000)
    	};
       \addplot [black, no markers, line width=1pt,dashed] coordinates {(-0.1,0.95) (1.1,0.95)};

        \legend{{$\threshold=0$,$\threshold=32$,$\threshold=64$,$\threshold=70$}}
        \end{axis}

        \end{tikzpicture}
\end{adjustbox}
        \caption{\label{fig:accuracy} Correctness~(with varying similarity threshold $
        \threshold$) and compression rate tradeoff.
        The dashed line marks $95\%$.}
\end{figure}
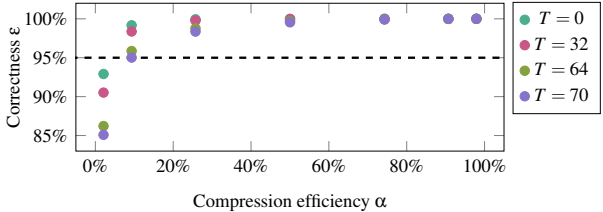

In Figure~\ref{fig:accuracy},
we present the trade-off between correctness and compression efficiency.
We plot the correctness $\epsilon$ (Y-axis) as the average compression rate
$\alpha$ varies (X-axis). 
The dashed horizontal line represents $95\%$.
We experiment with different definitions of similarity, i.e., with different values of $\threshold$, 
denoted by different colors.
Each node represents a parameter setup with the resulting correctness and
compression rate.
For example,
the second blue node from left in~\figref{fig:accuracy}
represents that with $\coarsethreshold=3$,
the resulting embedding scheme is $(32,98\%, \dist_{twitter})$-correct and $(\imageset, 9.3\%, \dist_{twitter})$-compact.
Note that our notion of correctness is stricter than the false negative rate defined in the prior work~\cite{kulshrestha2021identifying} for the Kulshrestha-Mayer protocol,
as the prior work measures the rate of transformed images that are not mapped to the original one,
constrained on the fact that the PDQHash of the transformed image stays in normalized Hamming distance $0.1$ to the original.
The protocol~\cite{kulshrestha2021identifying} reported an average false positive rate of
$16.8\%$ on a different dataset. 
Using our definition, this would be estimated as $(25.6, 83.2\%, \dist_{KM})$-correct.

In conclusion, these experiments suggest that for $\dist_{twitter}$, one
can achieve $\epsilonPrecR{\epsilonRecall > 0\%} < 50\%$, over $95\%$ correctness for the investigated values of
$\threshold$ and $9.3\%$ compression rate using $\EmbedScheme_{LSH}$ with $\emblen=9$, $\flipbias=0.05$ and
$\coarsethreshold=3$.
Hence, the resulting scheme achieves an almost order of magnitude reduction in the amount of data input to
a second-stage similarity protocol.

\begin{table}
\centering
\footnotesize
\begin{tabular}[t]{llr}
\toprule
Dataset &  Description & Images\\
\midrule
IMDB-WIKI~\cite{Rothe-ICCVW-2015,Rothe-IJCV-2018} & Faces. & 523,051\\
COCO~\cite{lin2014microsoft} & Common objects. & 123,403\\
T4SA~\cite{Vadicamo_2017_ICCVW} & Twitter images. & 1,473,394\\
Webvision 2.0~\cite{li2017webvision} & Flickr and Google images. & 13,907,566\\
\bottomrule
\end{tabular}

\caption{\label{tab:end-to-end-dataset} Dataset statistics and description.}
   \vspace{-0.4cm} 
\end{table}

\subsection{End-to-end Simulation}
\label{sec:end-to-end}

We now perform end-to-end simulation on varying sizes of blocklists
$\imageset$ to demonstrate the improvement of execution time and bandwidth 
for different similarity protocols combined with SBB.
For the experiments, we use parameters suggested in \secref{sec:correctness-expr},
i.e. $\emblen=9$, $\flipbias=0.05$ and
$\coarsethreshold=3$.

\begin{table*}
\centering
\fontsize{7}{9}\selectfont
\begin{tabular}[t]{rrrrrrrrr}
\toprule
  \multirow{2}{*}{$|\imageset|$} & \multicolumn{2}{c}{Execution Time~(s)} & \multicolumn{2}{c}{Total Bandwidth~(MiB)} &  \multicolumn{2}{c}{Execution Time~(s)} & \multicolumn{2}{c}{Total Bandwidth~(MiB)} \\
  & \multicolumn{1}{c}{No SBB} & \multicolumn{1}{c}{SBB} & \multicolumn{1}{c}{No SBB} & \multicolumn{1}{c}{SBB} &\multicolumn{1}{c}{No SBB} & \multicolumn{1}{c}{SBB} & \multicolumn{1}{c}{No SBB} &\multicolumn{1}{c}{SBB}\\
\midrule
&\multicolumn{4}{c}{\textbf{Similarity Embedding Retrieval}} & \multicolumn{4}{c}{\textbf{Secure Sketch}}\\
\midrule
$2^{18}$&$0.76$ $(0.00)$&$0.02$ $(0.00)$&$18.41$ $(0.00)$&$0.21$   $(0.01)$ & $1664.98$ $(\phantom{3}870.78)$&$2.77$ $(\phantom{1}0.17)$&$78.81$ $(0.36)$&$0.89$ $(0.05)$\\
$2^{19}$&$1.55$ $(0.01)$&$0.03$ $(0.01)$&$36.84$ $(0.02)$&$0.46$   $(0.15)$ & $5702.44$ $(3049.44)$&$6.03$ $(\phantom{1}0.41)$&$158.66$ $(2.39)$&$1.90$ $(0.14)$\\
  $2^{20}$&$3.09$ $(0.01)$&$0.06$ $(0.02)$&$73.65$ $(0.00)$&$0.83$   $(0.06)$ &\multicolumn{1}{c}{--} &  $12.09$ $(\phantom{1}0.99)$&\multicolumn{1}{c}{--}&$3.56$ $(0.26)$\\
  $2^{21}$&$6.17$ $(0.01)$&$0.10$ $(0.01)$&$147.31$ $(0.01)$&$1.69$   $(0.12)$ &\multicolumn{1}{c}{--} & $26.70$ $(\phantom{1}2.40)$&\multicolumn{1}{c}{--}&$7.18$ $(0.55)$\\
  $2^{22}$&$12.41$ $(0.07)$&$0.43$ $(0.99)$&$294.61$ $(0.03)$&$3.33$   $(0.25)$  & \multicolumn{1}{c}{--} & $61.99$ $(\phantom{1}6.63)$&\multicolumn{1}{c}{--}&$14.35$ $(1.17)$\\
  $2^{23}$& \multicolumn{1}{c}{--} &$0.40$ $(0.02)$& \multicolumn{1}{c}{--} &$6.70$ $(0.41)$ & \multicolumn{1}{c}{--} & $157.57$ $(14.41)$&\multicolumn{1}{c}{--}&$28.40$ $(1.68)$\\
\midrule
&\multicolumn{4}{c}{\textbf{CrypTen}} & \multicolumn{4}{c}{\textbf{EMP}}\\
\midrule
  $2^{13}$&$17.85$ $(0.18)$&$1.18$ $(0.31)$&$941.42$ $(0.37)$&$11.26$
  $(\phantom{22}1.64)$ & $13.14$ $(0.06)$&$0.21$ $(0.01)$ & $1520.63$ $(0.44)$ &$17.14$ $(\phantom{33}1.82)$ \\

  $2^{14}$&$36.65$ $(0.44)$&$1.33$ $(0.30)$&$1882.82$ $(0.87)$&$21.40$
  $(\phantom{22}2.14)$& $26.93$ $(0.98)$&$0.36$ $(0.03)$ & $3040.57$ $(1.39)$ &$34.52$ $(\phantom{33}4.12)$ \\

  $2^{15}$&$71.76$ $(0.07)$&$1.55$ $(0.06)$&$3766.55$ $(0.81)$&$42.43$
  $(\phantom{22}4.06)$& $54.61$ $(1.26)$&$0.67$ $(0.03)$ & $6079.48$ $(2.67)$ &$70.45$ $(\phantom{33}3.68)$  \\

$2^{16}$&$147.10$ $(0.66)$&$2.19$ $(0.08)$&$7534.29$ $(0.90)$&$83.77$ $(\phantom{22}6.96)$  &$117.52$ $(8.11)$&$1.21$ $(0.14)$ & $12158.00$ $(0.88)$ &$133.42$ $(\phantom{3}16.20)$ \\

  $2^{17}$&\multicolumn{1}{c}{--} &$3.55$ $(0.40)$&\multicolumn{1}{c}{--} &$167.70$ $(\phantom{2}14.69)$&\multicolumn{1}{c}{--} & $2.42$ $(0.13)$ & \multicolumn{1}{c}{--} &$275.52$ $(\phantom{3}15.32)$  \\

  $2^{18}$&\multicolumn{1}{c}{--} &$6.45$ $(0.40)$&\multicolumn{1}{c}{--} &$347.53$ $(\phantom{2}22.65)$&\multicolumn{1}{c}{--} & $4.79$ $(0.28)$ & \multicolumn{1}{c}{--} &$551.26$ $(\phantom{3}32.51)$  \\
  
  $2^{19}$&\multicolumn{1}{c}{--}&$13.77$ $(1.22)$&\multicolumn{1}{c}{--}&$695.30$ $(\phantom{2}45.49)$&\multicolumn{1}{c}{--}&$9.63$ $(0.52)$ & \multicolumn{1}{c}{--} &$1097.53$ $(\phantom{3}59.78)$ \\
  
  $2^{20}$&\multicolumn{1}{c}{--}&$27.89$ $(1.88)$&\multicolumn{1}{c}{--}&$1386.84$ $(107.27)$& \multicolumn{1}{c}{--}&$20.24$ $(1.26)$ & \multicolumn{1}{c}{--} &$2299.33$ $(143.63)$ \\

\midrule
\end{tabular}
\caption{\label{tab:end-to-end} Average time and bandwidth of similarity
protocols without and with SBB for four different similarity
testing protocols. Dashes (--) indicate when execution failed 
due to poor scaling. Numbers in parentheses are standard deviations.}
\end{table*}

\textbf{Datasets.}
To form varying sizes of blocklist $\imageset$, we randomly sample images from
the datasets that were used in prior
work~\cite{kulshrestha2021identifying}.\footnote{These are
unfortunately not suitable
for the privacy simulations of previous sections; they don't include information about sharing
frequency.}
The details of the datasets are listed in \tabref{tab:end-to-end-dataset}.
In particular, for each experiment, we generate $\imageset$ of the
requisite size by uniformly selecting images (without replacement) from the union of the 
COCO, T4SA, and Webvision~2.0 datasets.
For client requests,
we sample half of the requests from the IMDB-WIKI dataset to simulate requests that don't match any image in the block list,
and the rest from the generated~$\imageset$ to simulate client requests that return a match.

For each set of experiments with a randomly generated $\imageset$ tested 
with similarity embedding retrieval (the server simply sends the embeddings of
bucket entries to the client),
we provide measurements over $40$ iterations.
With secure sketch,
we use $20$ iterations and with 2PC protocols,
we use $10$ iterations because these take significantly longer 
to run.

\textbf{Implementation.} We use an AWS EC2 p2.xlarge\footnote{\url{https://aws.amazon.com/ec2/instance-types/p2/}}
instance with 61\,GiB of memory and 200\,GiB disk storage for the server side computation.  The instance is
initialized with the deep learning AMI provided by Amazon.  An AWS EC2
t2.small\footnote{\url{https://aws.amazon.com/ec2/instance-types/t2/}} instance with 2\,GiB of memory and 64\,GiB storage in the same region acted as a client.
The measured bandwidth between the two instances was 1\,Gbits/sec for both
directions and the network latency was $0.9$\,ms.
The server side implementation uses Python and 
is parallelized using the GPU.  The client side implementation uses Go.  
For the secure sketch protocol, 
we use an oblivious pseudorandom function~(OPRF),
implemented by the circl Go library from CloudFlare~\cite{circl}.

For the 2PC protocols our setup is identical except that we use an AWS EC2
t2.large instance with 8\,GiB of RAM as the client to be able to handle the
2PC frameworks.  The bandwidth was measured to be 1.01\,Gbits/sec from server to
client, 721\,Mbits/sec from client to server. The observed network latency was
$0.3$\,ms.  For both the CrypTen and EMP frameworks, we used the computed
functionality that checks if there exists a hash among the server's input that
has Hamming distance less than $\threshold$ to the client-provided hash.
The server's input is the entire $\imageset$, and the generated SBB bucket,
in the non-bucketized and bucketized setting, respectively. We further
XOR the output of this comparison with a randomly-generated client-provided bit
so that only the client learns the result. For CrypTen, for simplicity we
configured the trusted-third party for Beaver triple generation to run on the
same EC2 instance as the client (so-called trusted first-party mode). This is
not a secure configuration but provides lower bounds on performance (moving
Beaver generation to another server would decrease performance).  Experimental
results for CrypTen should therefore be considered to be lower bounds on
performance for secure deployments.
Note that our 2PC prototypes are not optimized, and absolute timings would be
improved using custom protocols for our setting such as the Kulshrestha-Mayer protocol~\cite{kulshrestha2021identifying}.
However, it is unclear if the protocol can be combined with SBB since it requires generating all buckets at setup time.

\textbf{Results.}
We present the average total execution time, average total bandwidth, and the speedup
provided by SBB for varying $|\imageset|$
in \tabref{tab:end-to-end}. 
The execution time and total bandwidth do not vary much between client
requests that match images in $\imageset$ and those that do not.
Many of the similarity testing protocols do not scale well, and we denote by dashes
in the table experiments that failed to complete. Typically this was due to
the client instance running out of memory. In all these cases, SBB was able to
increase scaling to complete executions with the available resources.

Our results show that SBB drastically improves the similarity protocol's
performance, both in terms of execution time and total bandwidth.  For
similarity embedding retrieval, SBB provides a $29\times$
speedup~($|\imageset|=2^{22}$) in execution time.  For large-scale
datasets~($|\imageset|\leq 2^{23}$), similarity embedding retrieval with SBB 
returns the answer in real time, under $0.5$ seconds.  For the secure
sketch protocol,
the improvement is even larger,  the speedup is at least $601\times$ for
$|\imageset|=2^{18}$.  For 2PC protocols, the improvement provided by SBB grows
larger as $\imageset$ becomes bigger, when $|\imageset|=2^{16}$, the speedup in
execution time is $67\times$ and $97\times$ for CrypTen and EMP, respectively.
For $|\imageset|=2^{20}$, EMP with SBB takes less time than the Kulshrestha-Mayer protocol~($20.24$s vs $27.5$s), however it requires larger bandwidth.

\section{Related Work}
\label{sec:related}

\textbf{Secure proximity search.}
Secure proximity search based on general multi-party computation has been used in many applications,
including privacy-preserving facial recognition~\cite{erkin2009privacy,sadeghi2009efficient},
biometric authentication~\cite{barni2010privacy,evans2011efficient}, querying sensitive health data~\cite{asharov2018privacy}, and more.
However,
these works don't scale sufficiently for our use case.
More scalable solutions include fuzzy extractors ~\cite{dodis2004fuzzy,juels1999fuzzy},
which give a small piece of client information to the server,
that does not leak information about the secret,
in order to derive secrets from noisy readings such as biometrics.
However, the security guarantee is based on the assumption that the distribution of the secret~(for example, fingerprints) has enough minimum entropy,
which is not necessarily true for image distributions.

\noindent\textbf{Private information retrieval.}
Chor et al.~\cite{chor1995private} introduced the concept of private information retrieval,
a type of protocol that enables a client to retrieve an item from a database such that the identity of the item is not revealed to the server.
The protocol requires clients to supply the index of the data item that they are querying about for retrieval purposes.
However, this is unlikely to be applied to our use case.
The most likely solution for similarity lookup is content-based privacy preserving retrieval~\cite{shashank2008private,xia2016privacy,lu2009enabling}.
Previous works~\cite{xia2016privacy,lu2009enabling} in this area assume three parties involved in this type of protocol,
a data owner, the client who queries the service and an actual server where the service is hosted.
The data owner encodes images or other types of multimedia into feature vectors that are further encoded with searchable encryption schemes and used as indices.
The server is made oblivious to the actual content of indices and the content of corresponding data items.
The client queries the server by generating indices from their images and receives the data items as answers.
The threat model doesn't prevent the case when the data owner colludes with the server,
hence cannot be directly applied to our scenario, 
where the data owner is the server.

\noindent
\textbf{Secure k nearest neighbors search.}
Another possible solution is to use secure k nearest neighbors search~(k-NNS) to look for similar items at the server side without revealing client information.
Most of the k-NNS solutions require a linear scan of the database that is queried against.
Recent work~\cite{chen2020sanns} proposed a sub-linear clustering-based algorithm,
yet the solution requires significant preprocessing time for each client.
Similar to our work,
Riazi et al.~\cite{riazi2016sub} used locality sensitive hashing to encode client queries for efficient k-NNS.
Different from our approach, they preserved user privacy by converting the LSH encodings to secure bits.
However, the notion of security is different in their work,
in particular,
an adversary can estimate the similarity of two given data points based on the encodings. 
This makes the protocol vulnerable to the matching attack,
which we addressed in our work.

\noindent
\textbf{Location-based services.} We draw comparisons between our bucketization setting and that of location-based services (LBSs). Andrés et al. \cite{geoind} propose a privacy-preserving mechanism in which mobile clients send their noise-perturbed locations to a server in order to obtain recommendations for nearby restaurants. One may view the noise-perturbed location as a coarse embedding and the server-provided list of restaurants as a similarity bucket. Similar to our coarse embedding scheme, the mechanism of Andrés et al. suffers from privacy loss when applied repeatedly to the same user input.
These connections suggest that our framework could have applications to reasoning about LBS privacy.
Conversely, insights from location privacy may serve as inspiration for improved SBB mechanisms.

\noindent
\textbf{Privacy measures.}
Prior work has proposed measuring privacy using an adversary's expected error when making inferences based on a posterior distribution on user inputs \cite{quantifyingLocPriv, isGeoIndWhat, optLocPriv}. Recent work has explored the Bayes security measure \cite{bayesMeasure}, which is similar to $\epsilonAcc$, but involves a security game in which the adversary attempts to recover a secret input as opposed to guessing a predicate on the secret input. Local differential privacy \cite{dwork2008differential} has also proven to be a popular worst-case privacy measure, but often incurs high correctness penalties.
Although similar, these metrics can not be directly applied to our scenario nor replace $\epsilonAuc$ and $\epsilonPrecR{\epsilonRecall > \epsilonRecallThresh}$.

\section{Limitations}
\label{sec:limitations}

Our work naturally suffers from several limitations that should be explored
further before deployments are considered.  Most notably, use of an SBB
mechanism fundamentally must leak some information to the server to trade-off
client privacy for efficiency.  In some contexts leaking even a single bit of
information about user content would be detrimental, in which case our
techniques are insufficiently private.  We speculate that leaking some
information about client images is, however, fundamental to achieve practical
performance in deployment for large~$\imageset$.  How to provide a formal
treatment establishing that scaling requires some leakage and what that means
for moderation mechanisms remain open questions.

Second, our empirical analyses focus on matching attacks for a single query,
which excludes some other potential threats. In particular it does not address
adversarially-known correlations between multiple images queried by one or more
clients.  A simple example, mentioned in \secref{sec:sbb}, is an `averaging'
attack against our LSH-based coarse embedding in which the adversary obtains a
large number of embeddings all for the same image~$\image$. Then the adversary
can average out the per-bit noise and recover the granular embedding
$\embed(\image)$.  We discuss simulation results for this scenario in \apref{ap:repeated}.  
The results indicate that,
similar to privacy-preserving mechanisms for location-based
services~\cite{geoind,quantifyingLocPriv, isGeoIndWhat,bayesMeasure}, repeated
queries on the same content drastically weaken the privacy guarantee of SBB:
an adversary that sees multiple SBB outputs that it knows are for the same
image can obtain near-perfect
matching  attack precision for almost all recall thresholds.

To address risk here, client software might cache images they've recently queried.
The client would not query the similarity service if a 
new image is too close to a prior image, and instead just reuse the cached result for
the latter.  Caching may not be feasible in all cases, and doesn't speak to
cross-user sharing of images, which may be inferrable from traffic analysis
should the adversary have access both to the similarity service and the
messaging platform. Another approach 
would be to somehow ensure that the same noise is added to the same
image, regardless of which client is sending it. This could possibly be done by
having some clients share a secret key, and use it to apply a 
pseudorandom function to the image (or its PDQHash) to derive the coins needed
for the random choices underlying our coarse embedding. 
Here an adversary's  $\epsilonPrecR{\epsilonRecall > 0\%}$ advantage remains at $50\%$ regardless of
the increasing number of repeated queries. But this doesn't account for other
potentially adversarially-known correlations across images (e.g., they are
almost identical), and may be fragile in the face of malicious client
software. Moreover, sending the same SBB embedding for the same image
would seem to increase susceptability to linking attacks in which
an adversary infers when two or more queries correspond to the same
image. We are unsure which scenario bears more risk in practice. 
We leave the exploration of these mitigations to future work.

A related limitation is the exclusive use of empiricism for evaluation. While we focus on
Bayes-optimal adversaries, it would be preferable to couple empiricism with
analyses providing bounds on adversarial success.  
While our definitional
framework provides the basis for proving bounds on, e.g.,
precision for particular data distributions, 
we do not yet have proofs and it appears to be challenging.  We emphasize that such results
cannot fully replace empirical work, because even formal results would
necessarily make assumptions about data that must be empirically validated.
Nevertheless, we consider the empirical results presented in this initial work
as a proof-of-concept of the SBB framework and encourage future works to further
examine the theoretical bounds for this approach. 

Finally, the public perceptual hash algorithms that SBB relies on
increases the risk of evasion attacks that seek to modify images just enough to avoid
detection. This risk seems particularly acute
when using a similarity testing protocol that sends a bucket
of PDQHash values to the client, as the adversary could extract 
these values from a client to inform their attacks.
Allowing users to report
misinformation images to frequently update the database may mitigate this risk.

\section{Conclusion and Future Directions}
\label{sec:conclusion}

In this paper, in order to allow efficient privacy-preserving similarity
testing, we defined the framework of similarity-based bucketization and
formalized a set of privacy goals that are important to this application.  We
consider the information that the adversary wants to infer from a client input
as the answer to a prediction task.  An adversary's advantage is measured by
their uncertainty regarding the prediction, using metrics that are widely
applied in machine learning.

Towards a realistic prototype for SBB, we focus on image similarity testing.
Driven by the privacy formalization, we ran simulations on real-world social
media data and analyzed the SBB protocol's security against a ``matching
attack''.  The attack refers to the scenario where an adversary tries to infer
if a client input is similar to an adversary-chosen image.  Using our framework,
deployments can tune the performance/privacy trade-off depending on the
application context.  We then test SBB's performance when composed with four
similarity protocols with varying server privacy guarantees for the server
content.  We show that the composition with SBB significantly reduces the
protocol latency and required bandwidth. While further research is needed to
address various open questions and limitations of our results,
we nevertheless believe that SBB represents a promising approach to
scaling private similarity testing in practice.

\section{Acknowledgements}

We thank Lucas Dixon, Nirvan Tyagi, and Paul
  Grubbs for useful discussions and
  feedback about this work. We
thank Laurens van der Maaten and Brian Knott for
answering our questions regarding CrypTen, and Xiao Wang for answering questions
about EMP.  This research was supported in part by NSF award \#1704527.

\bibliographystyle{plain}
\bibliography{reference}

\appendix

\section{A Secure Sketch Similarity Protocol}
\label{sec:sssp}

Here we provide further details regarding our secure sketch based protocol. 
A secure sketch~(SS) is a
cryptographic mechanism originally suggested for use with biometrics or other
fuzzy secrets.
Formally, an SS is a pair of algorithms $(\Sketch,\Rec)$. The
first is randomized, takes as input (in our context) a similarity representation $\simimage=\embed(\image)$ computed from image $\image$, and
outputs a bit string $\sketchval$, called a sketch. The recovery algorithm $\Rec$ takes as input a
sketch $\sketchval$ and a value $\simimage'$ and outputs a corrected value~$\simimage''$.
Informally, correctness requires that $\simimage' = \simimage''$ 
if $\Delta(\simimage,\simimage') \le \threshold$. 
For our purpose, 
we use secure sketches that work on the output space of similarity embedding $\simhash$ for $\Delta$ being Hamming distance. 

To realize our secure-sketch similarity protocol~(SSSP),
we use an oblivious pseudorandom function~(OPRF),
implemented by the circl Go library from CloudFlare~\cite{circl}.
It suffices to have a PRF $F_K(X)$ whose output is a member of a group,
and for which $F_K(X)^{K'} = F_{K\cdot K'}(X)$ for all $K,K',X$,
and where the keys form a group with
operation `$\cdot$'.  
The SSSP starts by having the
server compute a secure sketch $\sketchval$ for each image $\image' \in \bucketset$ using the Reed-Muller code~\cite{reed1953class,muller1954application},
as well
as the precomputed OPRF output for each such $\image'$. The resulting values
$\sketchval_1,\ldots,\sketchval_{|\bucketset|}$ and
$F_K(\simhash(\image_1')),\ldots,F_K(\simhash(\image_{|\bucketset|}'))$ are sent to the client. The
client chooses a random OPRF blinding key~$\tau$ and 
computes $\prfval_i \gets F_{\tau}(\Rec(\sketchval_i,\simhash(\image)))$ for each sketch in
the bucket. It sends the resulting $\prfval_1,\ldots,\prfval_{|\bucketset|}$ to
the server.
(If there are collisions in recovered similarity hashes,
the client replaces the repeat OPRF output with a random value.) 
The server computes and returns to the client
$\prfval_1^K,\ldots,\prfval_{|\bucketset|}^K$. Finally, the client can unblind
the OPRF values by raising them each to $\tau^{-1}$, and then looks for values
that match one of $\prfval_1,\ldots,\prfval_{|\bucketset|}$. 
The secure sketch protocol was implemented as a Go app at both the server side and client side.

\paragraph{Background on secure sketches.}
Our description of secure sketches follows~\cite{dodis2004fuzzy}.
Let $V$ be a random variable. The min-entropy of $V$ measures the probability of
correctly guessing a draw of~$V$ in a single guess. Formally it is defined as
\bnm
  \minentropy(V) = -\log \max_{v} \Prob{V = v} \;.
\enm
where the maximization is over all values $v$ in the support of~$V$.
We also can define the conditional min-entropy of a variable $V$ conditioned on
another (usually correlated) value~$Z$. Formally we define it as
\bnm
  \condminentropy{V}{Z} = -\log \sum_{z} \max_v \CondProb{V = v}{Z = z}\cdot\Prob{Z=z}
\enm
where the summation is over all values $z$ in the support of $Z$ and the
maximization  over all values $v$ in the support of $V$. In words this is just
the negative log of the expected min-entropy of $V$  conditioned on $Z$,
averaged over choice of $Z$. 

An $(\simimagespace,\minent,\redminent,\threshold)$-secure sketch is a pair of
algorithms $(\Sketch,\Rec)$ for which:
\begin{newitemize}
\item The sketching algorithm $\Sketch$ takes as input $\simimage \in \simimagespace$ and
outputs a string $\sketchval \in \bits^*$. 
\item The recovery algorithm $\Rec$ takes as input a string $\sketchval$ and an
image similarity representation $\simimage'$, and outputs a corrected value $\simimage''$.  Correctness mandates
that if $\Delta(\simimage,\simimage') \le \threshold$ then $\simimage'' = \simimage$. 
Otherwise no guarantees about $\simimage''$ are made. 
\item The following security guarantee holds. For any distribution $\dist$ over
$\simimagespace$ that has min-entropy $\minent$, then it holds that 
$\condminentropy{V}{\Sketch(V)} \ge \redminent$ where $V$ is the random variable
representing a value drawn according to $\dist$ and $\Sketch(V)$. 
\end{newitemize}

For our context, the security guarantee means that if images have high
min-entropy then the secure sketch ensures that they remain unpredictable.

\paragraph{Oblivious PRFs.} In addition to a secure sketch, our protocol also
relies on one other cryptographic object, an oblivious PRF (OPRF)~\cite{jarecki2009efficient,jarecki2014round}. An
OPRF allows one party with a message $\embimage$ to obtain the output of a 
PRF $F_K(\embimage)$ where $K$ is held by the other party. We particularly use
an OPRF construction from~\cite{jarecki2014round}, which is based in part off
earlier protocols due to Ford and Kaliski~\cite{ford2000server}.  The OPRF uses a
cryptographic group (such as a suitable elliptic curve group) $\G$ of order~$p$ with
generator $g$, a hash function $H\Colon\bits^*\rightarrow\G$, and a hash
function $H'\Colon\G\rightarrow\bits^n$ for some parameter~$n$ (e.g., $n = 256$
if one uses SHA-256 for $H'$). The  party holding input $\embimage$ begins by choosing a
random integer $\tau\getsr \Z_p$, computes $\blindval \gets H(\embimage)^\tau$ and sends the
result to the party holding the secret key $K$ (chosen uniformly 
from~$\Z_p$). Upon receiving $Y$, that party computes $\gets \blindval^K$ and sends
$Y$ back. Finally, the initiator computes $Y \gets \blindval^{\tau^{-1}}$ and outputs
$H'(\embimage,Y)$. Thus ultimately the PRF is defined as 
$F_K(\embimage) = H'(\embimage,H(\embimage)^K)$.

We note that this scheme can be extended to allow verification that the server
consistently uses the same PRF key $K$. We refer the reader
to~\cite{jarecki2014round} for details.

\paragraph{The protocol.} Our protocol involves a client whose input is an image similarity representation
$\simimage \in \simimagespace$ and a server whose input is a set of images
$\imageset \subseteq \imagespace$. The protocol is parameterized by a
similarity embedding function $\embed$ and distance threshold $\threshold$. We
assume a$(\simimagespace,\minent,\redminent,\threshold)$-secure sketch
$(\Sketch,\Rec)$ that works for the distance measure $\Delta$ associated with
$\embed$. We also use the OPRF protocol described above.  

A diagram of the protocol appears in \figref{fig:sssp}. 
Our secure-sketch similarity protocol (SSSP) starts by having the
server compute a secure sketch for each image $\image' \in \bucketset$, as well
as the OPRF output for each such $\image'$. The resulting values
$\sketchval_1,\ldots,\sketchval_{|\bucketset|}$ and
$F_K(\simhash(\image_1')),\ldots,F_K(\simhash(\image_{|\bucketset|}'))$ are sent to the client. The
client chooses a random OPRF blinding key~$\tau$ and 
computes $\prfval_i \gets F_{\tau}(\Rec(\sketchval_i,\simhash(\image)))$ for each sketch in
the bucket. It sends the resulting $\prfval_1,\ldots,\prfval_{|\bucketset|}$ to
the server.
(If there are collisions in recovered similarity hashes,
the client replaces the repeat OPRF output with a random value.) 
The server computes and returns to the client
$\prfval_1^K,\ldots,\prfval_{|\bucketset|}^K$. Finally, the client can unblind
the OPRF values by raising them each to $\tau^{-1}$, and then looks for values
that match one of $\prfval_1,\ldots,\prfval_{|\bucketset|}$. 

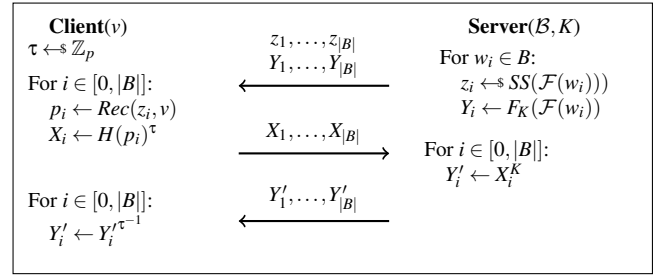
\begin{figure}[t]
\center
{\footnotesize
\begin{tikzpicture}[yscale=-1,node distance=0.5cm]
    \coordinate(topleft) at (0.0,0.0);
    \coordinate(clientcoord)  at (1.0,0.10);

    \draw (topleft) rectangle ++(8.5,3.6); 

    \node(client)[align=center,minimum width=3,anchor=north] at (clientcoord) 
          {\textbf{Client}($\simimage$)};
    \node(server)[align=center,minimum width=3,anchor=north] at ($(clientcoord)+(5.8,0)$) 
    {\textbf{Server}($\imageset,K$)};
    
     \node(clientline) [align=left,minimum width=3,anchor=north west] at (0.1,0.4)
          {
    $\tau\getsr \Z_p$
          };

    \node(serverline)[align=left,minimum width=3, anchor=north] at ($(server)+(0,.20)$)
          {
            For $\image_i\in\bucketset$:\\ 
            \myInd $\sketchval_i \getsr \Sketch(\embed(\image_i)))$\\
            \myInd $\prfval_i \gets F_K(\embed(\image_i))$
          };

  \draw[thick,<-] (3,1.10) -- node [text width=3cm,midway,above,align=center] 
        {$\sketchval_1,\ldots,\sketchval_{|\bucketset|}$\\
         $\prfval_1,\ldots,\prfval_{|\bucketset|}$
        } ++(2,0);

    \node(clientline) [align=left,minimum width=3,anchor=north west] at (0.1,0.85)
          {
            For $i\in[0,|\bucketset|]$:\\ 
            \myInd $\embimage_i \gets \Rec(\sketchval_i,\simimage)$\\
            \myInd $\blindval_i \gets H(\embimage_i)^\tau$\\
          };
   
  \draw[thick,->] (3,2) -- node [text width=3cm,midway,above,align=center] 
        {$\blindval_1,\ldots,\blindval_{|\bucketset|}$
        } ++(2,0);

    \node(serverline2)[align=left,minimum width=3,anchor=north] at ($(server)+(-0.5,1.4)$)
          {
            For $i\in[0,|\bucketset|]$:\\ 
            \myInd $\blindprfval_i \gets \blindval_i^K$
          };
    
  \draw[thick,<-] (3,2.9) -- node [text width=3cm,midway,above,align=center] 
        {$\blindprfval_1,\ldots,\blindprfval_{|\bucketset|}$
        } ++(2,0);

    \node(clientline2) [align=left,minimum width=3,anchor=north west] at (0.1,2.4)
          {
            For $i\in[0,|\bucketset|]$:\\ 
            \myInd $\prfval'_i \gets {{{\blindprfval_i}^{\tau^{-1}}}}$\\
            
          };

\end{tikzpicture}
\vspace{-0.7cm}
}
\caption{\label{fig:sssp}$F_{(\cdot)}(\cdot)$ is a PRF. $SS(\cdot)$  and $Rec(\cdot, \cdot)$ are from the secure sketch protocol.}
\end{figure}

\section{Correctness bound of $\EmbedScheme_{LSH}$}
\label{sec:correctness}

\begin{theorem}
For any $\image, \image' \in \imagespace$ with $\distimage(\image,\image')<\threshold$,
when applying $\EmbedScheme_{LSH}$ with embedding length $\emblen$, flipping bias $\flipbias$ and coarse threshold $\coarsethreshold$,
for any $\beta > 1$ and $\coarsethreshold > \emblen\frac{\threshold + \beta \len \flipbias}{\len}$,
we have 
\begin{equation*}
\begin{multlined}
\prob[\Sim_{LSH}(\Embed_{LSH}(\image),\image') = \true]\\
\indent > (1- e^{-2\len(\beta-1)^2\flipbias^2})
\cdot(1-e^{-2\emblen(\frac{\coarsethreshold}{\emblen}-\frac{\threshold + \beta \len \flipbias}{\len})^2})\;.
\end{multlined}
\end{equation*}
\end{theorem}

\begin{proof}
$\Embed_{LSH}$ works as follows: 
First, the algorithm samples an index function family $\indexfunc$ as 
a combination of $\emblen$ functions sampled uniformly from $\indexfuncset$ without replacement. 
The function $\indexfunc$ is then applied on the similarity embedding of $\image$,
hence we have $\embimage=\indexfunc(\embed(\image))$.
A randomized algorithm $\flip_\flipbias$ that takes as input
a bit string $\embimage$ and outputs $\pnoise$ of the same length,
setting $\pnoise_i = \embimage_i$ with probability
$1 - \flipbias$ and $\pnoise_i = \lnot \embimage_i$ with probability $\flipbias$. 
A threshold $\coarsethreshold$ defined with the embedding scheme 
is used as the coarse threshold for
the Hamming distance over the randomly selected indexes.

Swapping the order of applying $\flip_{\flipbias}$ and $\indexfunc$ doesn't change the output of $\Embed(\image)$.
Hence, 
we apply $\flip_\flipbias$ on $\embed(\image)$ first.
Let $\simimage=\embed(\image)$ 
and $\vnoise=\flip_\flipbias(\embimage)$.
We define $X=\Sigma_{i=1}^\len X_i$, where the random variable $X_i=1$ if and only if $\vnoise_i$ is set to $\lnot \simimage_i$.
Since each index is flipped independently with probability $\flipbias$ as in a Bernoulli trial,
we have that $X$ is sampled from a binomial distribution $X \sim B(\len, \flipbias)$, where $\e[X]=\len\flipbias$.
Thus according to Hoeffding's inequality, for any $\delta > \len\flipbias$,
$\prob[X \geq \delta] \leq e^{-2\len(\frac{\delta}{\len}-\flipbias)^2}\;$.

For any $\image,\image' \in \imagespace$  with $\distimage(\image,\image') < \threshold$,
let $\simimage=\embed(\image)$, $\vnoise=\flip_\flipbias(\simimage)$ and $\simimage'=\embed(\image')$.
We have $\Delta(\vnoise,\simimage') \leq X+\Delta(\simimage,\simimage')$,
$X$ as defined previously.
Hence, in the case of $\Delta(\simimage,\simimage') < \threshold$, for self-defined $\beta$, such that $\beta> 1$,
\begin{align*}
\prob[\Delta(\vnoise, \simimage') \geq \threshold + \beta \len\flipbias] &\leq \prob[X \geq \threshold - \Delta(\simimage, \simimage') + \beta \len \flipbias]\\
&\leq e^{-2\len(\frac{(\threshold-\Delta(\simimage,\simimage')}{\len} + (\beta -1 )\flipbias)^2}\\
&\leq e^{-2\len(\beta -1)^2\flipbias^2}
\end{align*}

Consider a randomly chosen $\indexfunc$ from $\indexfuncset$,
we define a random variable $Y$ as the number of indexing functions being drawn, such that $\vnoise[\indexfunc] \neq \simimage'[\indexfunc]$.
Note that $Y=\Delta(\indexfunc(\vnoise), \indexfunc(\simimage')))$,
and it is drawn from a hypergeometric distribution.

Thus, with Hoeffding's inequality, we have 
$\prob[Y \geq \coarsethreshold] \leq e^{-2\emblen(\frac{\coarsethreshold}{\emblen}-\frac{\Delta(\vnoise, \simimage')}{\len})^2}\;$.
In the case where $\coarsethreshold > \emblen \cdot \frac{\threshold+ \beta \len \flipbias}{\len}$, 
$\prob[\Delta(\indexfunc(\vnoise), \indexfunc(\simimage')) \geq \coarsethreshold \mid \Delta(\vnoise, \simimage') < \threshold + \beta \len\flipbias] \leq e^{-2\emblen(\frac{\coarsethreshold}{\emblen}-\frac{\threshold + \beta \len \flipbias}{\len})^2}\;$.
Hence, we have,
\begin{align*}
\prob[\Sim_{LSH}&(\Embed_{LSH}(\image),\image') = \true]\\
&= \prob[\Delta(\Embed_{LSH}(\image),\Embed_{LSH}(\image')) < \coarsethreshold] \\
&> \prob[\Delta(\vnoise, \simimage') < \threshold + \beta \len\flipbias]\\
&\qquad\cdot \prob[\Delta(\indexfunc(\vnoise), \indexfunc(\simimage')) < \coarsethreshold \mid \Delta(\vnoise, \simimage') < \threshold + \beta \len\flipbias]\\
&> (1- e^{-2\len(\beta-1)^2\flipbias^2})(1-e^{-2\emblen(\frac{\coarsethreshold}{\emblen}-\frac{\threshold + \beta \len \flipbias}{\len})^2})\;.
\end{align*}

\end{proof}
\section{Probability Computed by Bayes Optimal Adversary}
\label{ap:bayes}

We prove Theorem~\ref{thm:bayes-opt} in this section and recall it below.
\bayesOpt*

\begin{proof}
By Bayes' theorem,
\begin{equation*}
\begin{multlined}
\Pr[\simimage = \simimage_{adv} \mid \Embed(\image) = (I, p)] \\
= \dfrac{\Pr[\Embed(\image) = (I, p) \mid \simimage = \simimage_{adv}] \cdot \Pr[\simimage = \simimage_{adv}]}{\Pr[\Embed(\image) = (I, p)]}
\end{multlined}
\end{equation*}

Observe that $\Pr[\simimage = \simimage_{adv}] = \distf(\simimage_{adv})$ and that

\begin{align*}
    \Pr&[\Embed(\image) = (I, p) \mid \simimage = \simimage_{adv}] \\
    &= \Pr[\Embed(\image) = (I, p) \mid I, \simimage = \simimage_{adv}] \Pr[I]\\
    &= \gamma^{\Delta_I(\simimage_{adv}, p)} \cdot (1 - \gamma)^{d - \Delta_I(\simimage_{adv}, p)} \cdot \Pr[I]
\end{align*}
The last equality holds by the independence of the $\gamma$-biased coin flips, with $(1 - \gamma)$ corresponding to the bits in common and $\gamma$ corresponding to the differing bits.

Now, note that $$\Pr[\Embed(\image) = (I, p)] = \Pr[\Embed(\image) = (I, p) \mid I] \Pr[I]\;.$$ Observe that

\begin{align*}
    \Pr&[\Embed(\image) = (I, p) \mid I] \\
    &= \sum\limits_{\simimage' \in \bits^\len} \Pr[\Embed(\image) = (I, p) \mid I, \simhash(\image) = \simimage'] \cdot \Pr[\simhash(\image) = \simimage']\\
    &= \sum\limits_{\simimage' \in \bits^\len} \gamma^{\Delta_I(\simimage', p)} \cdot (1 - \gamma)^{d - \Delta_I(\simimage', p)}\cdot \distf(\simimage')
\end{align*}

Plugging back into our original expression and cancelling the $\Pr[I]$ terms completes the proof of the theorem.
\end{proof}

\begin{figure*}[h!]
\centering
\begin{adjustbox}{width=\textwidth}
   \begin{tikzpicture}
 \begin{axis}[
    name=plot1,
    xbar stacked,
	bar width=9pt,
	y=12pt,
	ymin=-0.5,
	ymax=3.5,
	xmin=0,
	xmax=1,
	height=3.5cm,
    width=0.44\textwidth,
    enlarge y limits=0.1,
    legend style={at={(1.15,1)},
      anchor=north,legend columns=1, font=\footnotesize},
    xlabel={Percentage of Similarity Embeddings},
    xlabel style = {font=\footnotesize},
    ticklabel style = {font=\footnotesize},
    ytick={0, 1, 2, 3},
    yticklabels={{$\threshold=0$}, {$\threshold=32$},
    {$\threshold=64$}, {$\threshold=70$}},
    xticklabel={\pgfmathparse{\tick*100}\pgfmathprintnumber{\pgfmathresult}\%},
    ]

\addplot[fill=clr2_1!25, draw=white, xbar] plot coordinates {(0.7846850917990699,0) (0.6849025088946387,1) 
  (0.6380505244804184,2) (0.6241696893108372,3)};
\addlegendentry{$1$}

\addplot[fill=clr2_1!50, draw=white, xbar] plot coordinates {(0.2036748954571691,0) (0.2623815103489738,1) 
  (0.2842205310522534,2) (0.2900538677181603,3)};
\addlegendentry{$(1, 10]$}
 \addplot[fill=clr2_1!75, draw=white, xbar] plot coordinates {(0.011577360361837501,0) (0.04941806601428208,1) 
  (0.06879898049911287, 2) (0.07533482366020547,3)};
\addlegendentry{$(10, 100]$}
  \addplot[fill=clr2_1!100, draw=white, xbar] plot coordinates {(6.265238192358809e-05,0) (0.003297914742105467,1) 
  (0.008929963968215247,2) (0.010441619310797139,3)};
\addlegendentry{$> 100$}

 \addplot[fill=gray!0, opacity=0, draw=white, xbar, forget plot] plot coordinates {(-1,1) (-1,3) 
  (-1,5) (-1,7)};
 
\end{axis}
    \begin{axis}[
                      name=plot2,
            at=(plot1.right of south east), anchor=left of south west,
    	    legend columns = 1,
    	    ticklabel style = {font=\footnotesize},
    	    xlabel={Coarse Embedding Length $\emblen$},
    	    ylabel= {Adversarial advantage},
            ylabel style={font=\footnotesize, at={(axis description cs:0.08,.5)},anchor=south},
	        xshift = 0.3cm,
    	    xlabel style = {font=\footnotesize},
    	    xtick = {8,9,10,11,12,13,14,15,16,17.5},
            xticklabels = {$8$,$9$,$10$,$11$,$12$,$13$,$14$,$15$,$16$},
            xticklabel style={align=center},
            yticklabel={\pgfmathparse{\tick*100}\pgfmathprintnumber{\pgfmathresult}\%},
    		height=3.5cm,
    		width=0.44\textwidth,
    		scaled y ticks=false,
    		legend style={/tikz/every even column/.append style={column sep=0.5cm}, at={(1.25,1),font=\footnotesize},anchor=north east,name=legend}
    	]
        \addplot[clr3_1, only marks,error bars/.cd, y dir=both, y explicit] coordinates {
(8.0000, 0.9981) += (0, 0.0002) -= (0, 0.0002)
(9.0000, 0.9990) += (0, 0.0001) -= (0, 0.0001)
(10.0000, 0.9994) += (0, 0.0001) -= (0, 0.0001)
(11.0000, 0.9997) += (0, 0.0001) -= (0, 0.0001)
(12.0000, 0.9998) += (0, 0.0000) -= (0, 0.0000)
(13.0000, 0.9999) += (0, 0.0000) -= (0, 0.0000)
(14.0000, 1.0000) += (0, 0.0000) -= (0, 0.0000)
(15.0000, 1.0000) += (0, 0.0000) -= (0, 0.0000)
(16.0000, 1.0000) += (0, 0.0000) -= (0, 0.0000)
    	};  		
      \addplot[clr3_2, only marks,error bars/.cd, y dir=both, y explicit,] coordinates {
(8.0000, 0.0649) += (0, 0.0057) -= (0, 0.0057)
(9.0000, 0.1160) += (0, 0.0058) -= (0, 0.0058)
(10.0000, 0.1939) += (0, 0.0208) -= (0, 0.0208)
(11.0000, 0.3132) += (0, 0.0365) -= (0, 0.0365)
(12.0000, 0.4602) += (0, 0.0284) -= (0, 0.0284)
(13.0000, 0.6624) += (0, 0.0349) -= (0, 0.0349)
(14.0000, 0.7693) += (0, 0.0141) -= (0, 0.0141)
(15.0000, 0.8523) += (0, 0.0329) -= (0, 0.0329)
(16.0000, 0.8950) += (0, 0.0313) -= (0, 0.0313)

       };

    	\addplot[clr3_3, only marks,error bars/.cd, y dir=both, y explicit,] coordinates {

(8.0000, 0.0000) += (0, 0.0000) -= (0, 0.0000)
(9.0000, 0.0000) += (0, 0.0000) -= (0, 0.0000)
(10.0000, 0.0000) += (0, 0.0000) -= (0, 0.0000)
(11.0000, 0.0000) += (0, 0.0000) -= (0, 0.0000)
(12.0000, 0.0108) += (0, 0.0221) -= (0, 0.0108)
(13.0000, 0.4825) += (0, 0.0851) -= (0, 0.0851)
(14.0000, 0.6994) += (0, 0.0236) -= (0, 0.0236)
(15.0000, 0.8237) += (0, 0.0450) -= (0, 0.0450)
(16.0000, 0.8803) += (0, 0.0391) -= (0, 0.0391)

    	};
    	\addplot[clr3_1, only marks, mark=diamond*,mark options={scale=2}] coordinates {(17.5, 0.964756)};
    	\addplot[clr3_2, only marks, mark=diamond*,mark options={scale=2}] coordinates {(17.5, 0.965956)};
    	\addplot[clr3_3, only marks,mark=diamond*,mark options={scale=2}] coordinates {(17.5, 0.999980)};
    	\addplot[red, line width=1pt] coordinates {
    	(17, 0.85) (18, 0.85)};
    	\addplot[red, line width=1pt] coordinates {
    	(17, 0.85) (17, 1.15)};
    	\addplot[red, line width=1pt] coordinates {
    	(18, 0.85) (18, 1.15)};
    	\addplot[red, line width=1pt] coordinates {
    	(17, 1.15) (18, 1.15)};
    	\legend{$\epsilonAuc$, $\epsilonPrecR{\epsilonRecall=100\%}$, $\epsilonAcc$}
        \node[
            anchor=north west,
            align=left,
            font=\footnotesize,
        ] at (axis cs:16.5,.85)
        {$\EmbedScheme_{cPDQ}$\\$\emblen=16$};

    \end{axis}

     \begin{axis}[
            name=plot3,
            at=(plot1.below south west), anchor=above north west,
    	    legend columns = 1,
    	    ticklabel style = {font=\footnotesize},
    	    xmin=-0.05,
    	    xmax=0.4,
	        ymin=-0.05,
	        ymax=0.75,
            xticklabel style={align=center},
            yticklabel={\pgfmathparse{\tick*100}\pgfmathprintnumber{\pgfmathresult}\%},
            ylabel={$\epsilonPrecR{\epsilonRecall > \epsilonRecallThresh}$},
            ylabel style={font=\footnotesize, at={(axis description cs:0.08,.5)},anchor=south},
            xlabel={Flipping bias $\flipbias$},
            xlabel style = {font=\footnotesize}, 
    		height=3.5cm,
    		width=0.41\textwidth,
    		scaled y ticks=false,
    	    legend style={nodes={scale=.95, font=\footnotesize}}, 
    		legend style={/tikz/every even column/.append style={column sep=0.38cm}, at={(1.01,1)},anchor=north west,font=\footnotesize},
    	]

      \addplot[clr5_1,mark=*,line width=1pt,error bars/.cd, y dir=both, y explicit,] coordinates {
(0.0000, 0.4602) += (0, 0.0284) -= (0, 0.0284)
(0.0500, 0.3576) += (0, 0.0266) -= (0, 0.0266)
(0.1000, 0.2190) += (0, 0.0110) -= (0, 0.0110)
(0.1500, 0.1204) += (0, 0.0163) -= (0, 0.0163)
(0.2000, 0.0718) += (0, 0.0113) -= (0, 0.0113)
(0.2500, 0.0380) += (0, 0.0054) -= (0, 0.0054)
(0.3000, 0.0136) += (0, 0.0031) -= (0, 0.0031)
(0.3500, 0.0049) += (0, 0.0014) -= (0, 0.0014)
       };
       \addplot[clr5_2,mark=*,line width=1pt,error bars/.cd, y dir=both, y explicit,] coordinates {
(0.0000, 0.4602) += (0, 0.0284) -= (0, 0.0284)
(0.0500, 0.3576) += (0, 0.0266) -= (0, 0.0266)
(0.1000, 0.2063) += (0, 0.0283) -= (0, 0.0283)
(0.1500, 0.0439) += (0, 0.0047) -= (0, 0.0047)
(0.2000, 0.0217) += (0, 0.0027) -= (0, 0.0027)
(0.2500, 0.0070) += (0, 0.0007) -= (0, 0.0007)
(0.3000, 0.0033) += (0, 0.0003) -= (0, 0.0003)
(0.3500, 0.0014) += (0, 0.0001) -= (0, 0.0001)
       };
       \addplot[clr5_3,mark=*,line width=1pt,error bars/.cd, y dir=both, y explicit,] coordinates {
(0.0000, 0.4602) += (0, 0.0284) -= (0, 0.0284)
(0.0500, 0.3576) += (0, 0.0266) -= (0, 0.0266)
(0.1000, 0.0618) += (0, 0.0032) -= (0, 0.0032)
(0.1500, 0.0223) += (0, 0.0045) -= (0, 0.0045)
(0.2000, 0.0078) += (0, 0.0003) -= (0, 0.0003)
(0.2500, 0.0032) += (0, 0.0003) -= (0, 0.0003)
(0.3000, 0.0016) += (0, 0.0002) -= (0, 0.0002)
(0.3500, 0.0009) += (0, 0.0001) -= (0, 0.0001)

      };
  		
      \addplot[clr5_4,mark=*,line width=1pt,error bars/.cd, y dir=both, y explicit,] coordinates {
(0.0000, 0.4602) += (0, 0.0284) -= (0, 0.0284)
(0.0500, 0.0891) += (0, 0.0080) -= (0, 0.0080)
(0.1000, 0.0203) += (0, 0.0016) -= (0, 0.0016)
(0.1500, 0.0083) += (0, 0.0016) -= (0, 0.0016)
(0.2000, 0.0032) += (0, 0.0002) -= (0, 0.0002)
(0.2500, 0.0015) += (0, 0.0002) -= (0, 0.0002)
(0.3000, 0.0009) += (0, 0.0001) -= (0, 0.0001)
(0.3500, 0.0006) += (0, 0.0000) -= (0, 0.0000)

       };
       \addplot[clr5_5,mark=*,line width=1pt,error bars/.cd, y dir=both, y explicit,] coordinates {
(0.0000, 0.4602) += (0, 0.0284) -= (0, 0.0284)
(0.0500, 0.0031) += (0, 0.0012) -= (0, 0.0012)
(0.1000, 0.0009) += (0, 0.0003) -= (0, 0.0003)
(0.1500, 0.0005) += (0, 0.0001) -= (0, 0.0001)
(0.2000, 0.0004) += (0, 0.0001) -= (0, 0.0001)
(0.2500, 0.0003) += (0, 0.0000) -= (0, 0.0000)
(0.3000, 0.0003) += (0, 0.0000) -= (0, 0.0000)
(0.3500, 0.0003) += (0, 0.0000) -= (0, 0.0000)
       };
       \addplot [black, no markers, line width=1pt,dashed] coordinates {(-0.5,0.5) (0.6,0.5)};

       \legend{{$\epsilonRecall > 0\%$}, {$\epsilonRecall > 25\%$},
       {$\epsilonRecall > 50\%$},{$\epsilonRecall > 75\%$}, {$\epsilonRecall =100\%$}}
    	\end{axis}
    	
    	\begin{axis}[
             name=plot4,
            at=(plot2.below south west), anchor=above north west,
    	    ticklabel style = {font=\footnotesize},
    	    legend columns = 1,
    	    xmin=-0.05,
    	    xmax=1.05,
    	    ymin=0.85,
    	    ymax=1.02,
            yticklabel={\pgfmathparse{\tick*100}\pgfmathprintnumber{\pgfmathresult}\%},
            xticklabel={\pgfmathparse{\tick*100}\pgfmathprintnumber{\pgfmathresult}\%},
    		height=3.5cm,
            ylabel style={font=\footnotesize, at={(axis description cs:0.08,.5)},anchor=south},
    		width=0.43\textwidth,
    		scaled y ticks=false,
    		ylabel = {Correctness $\epsilon$},
    		xlabel={Compression rate $\alpha$},
    		xlabel style = {font=\footnotesize},
    		legend style={/tikz/every even column/.append style={column sep=0.38cm}, at={(1.02,1)},anchor=north west,font=\footnotesize},
    	]
        \addplot[clr5_1, only marks,error bars/.cd, y dir=both, x dir=both, y explicit, x explicit] coordinates {
        (0.0003, 0.5409) += (0.0000, 0.0010) -= (0.0000, 0.0010)
(0.0037, 0.8819) += (0.0000, 0.0006) -= (0.0000, 0.0006)
(0.0212, 0.9804) += (0.0000, 0.0003) -= (0.0000, 0.0003)
(0.0770, 0.9978) += (0.0000, 0.0001) -= (0.0000, 0.0001)
(0.1987, 0.9998) += (0.0000, 0.0000) -= (0.0000, 0.0000)
(0.3896, 1.0000) += (0.0000, 0.0000) -= (0.0000, 0.0000)
(0.6108, 1.0000) += (0.0000, 0.0000) -= (0.0000, 0.0000)
(0.8016, 1.0000) += (0.0000, 0.0000) -= (0.0000, 0.0000)
(0.9231, 1.0000) += (0.0000, 0.0000) -= (0.0000, 0.0000)
(0.9789, 1.0000) += (0.0000, 0.0000) -= (0.0000, 0.0000)
(0.9963, 1.0000) += (0.0000, 0.0000) -= (0.0000, 0.0000)

    	};  		
      \addplot[clr5_2, only marks,error bars/.cd, y dir=both, y explicit, x dir=both, x explicit] coordinates {
(0.0003, 0.5125) += (0.0000, 0.0010) -= (0.0000, 0.0010)
(0.0037, 0.8560) += (0.0000, 0.0007) -= (0.0000, 0.0007)
(0.0212, 0.9686) += (0.0000, 0.0003) -= (0.0000, 0.0003)
(0.0770, 0.9942) += (0.0000, 0.0001) -= (0.0000, 0.0001)
(0.1987, 0.9990) += (0.0000, 0.0000) -= (0.0000, 0.0000)
(0.3896, 0.9999) += (0.0000, 0.0000) -= (0.0000, 0.0000)
(0.6108, 1.0000) += (0.0000, 0.0000) -= (0.0000, 0.0000)
(0.8016, 1.0000) += (0.0000, 0.0000) -= (0.0000, 0.0000)
(0.9231, 1.0000) += (0.0000, 0.0000) -= (0.0000, 0.0000)
(0.9789, 1.0000) += (0.0000, 0.0000) -= (0.0000, 0.0000)
(0.9963, 1.0000) += (0.0000, 0.0000) -= (0.0000, 0.0000)
       };
    	\addplot[clr5_3, only marks,error bars/.cd, y dir=both, y explicit,] coordinates {
(0.0003, 0.4813) += (0.0000, 0.0009) -= (0.0000, 0.0009)
(0.0037, 0.8120) += (0.0000, 0.0007) -= (0.0000, 0.0007)
(0.0212, 0.9340) += (0.0000, 0.0004) -= (0.0000, 0.0004)
(0.0770, 0.9747) += (0.0000, 0.0002) -= (0.0000, 0.0002)
(0.1987, 0.9907) += (0.0000, 0.0001) -= (0.0000, 0.0001)
(0.3896, 0.9971) += (0.0000, 0.0001) -= (0.0000, 0.0001)
(0.6108, 0.9993) += (0.0000, 0.0000) -= (0.0000, 0.0000)
(0.8016, 0.9999) += (0.0000, 0.0000) -= (0.0000, 0.0000)
(0.9231, 1.0000) += (0.0000, 0.0000) -= (0.0000, 0.0000)
(0.9789, 1.0000) += (0.0000, 0.0000) -= (0.0000, 0.0000)
(0.9963, 1.0000) += (0.0000, 0.0000) -= (0.0000, 0.0000)
    	};
    	\addplot[clr5_4, only marks,error bars/.cd, y dir=both, y explicit,] coordinates {
(0.0003, 0.4731) += (0.0000, 0.0009) -= (0.0000, 0.0009)
(0.0037, 0.7990) += (0.0000, 0.0007) -= (0.0000, 0.0007)
(0.0212, 0.9219) += (0.0000, 0.0004) -= (0.0000, 0.0004)
(0.0770, 0.9663) += (0.0000, 0.0003) -= (0.0000, 0.0003)
(0.1987, 0.9861) += (0.0000, 0.0002) -= (0.0000, 0.0002)
(0.3896, 0.9952) += (0.0000, 0.0001) -= (0.0000, 0.0001)
(0.6108, 0.9986) += (0.0000, 0.0000) -= (0.0000, 0.0000)
(0.8016, 0.9997) += (0.0000, 0.0000) -= (0.0000, 0.0000)
(0.9231, 1.0000) += (0.0000, 0.0000) -= (0.0000, 0.0000)
(0.9789, 1.0000) += (0.0000, 0.0000) -= (0.0000, 0.0000)
(0.9963, 1.0000) += (0.0000, 0.0000) -= (0.0000, 0.0000)
    	};
       \addplot [black, no markers, line width=1pt,dashed] coordinates {(-0.05,0.95) (1.05,0.95)};

        \legend{{$\threshold=0$,$\threshold=32$,$\threshold=64$,$\threshold=70$}};
        \end{axis}

    \end{tikzpicture}
    \end{adjustbox}
    \vspace{-0.7cm}
    \caption{\label{fig:overview-ap} \textbf{Upper-left}: $\threshold$-Neighborhood size distribution. \textbf{Upper-right}: Simulation results of the three security metrics against matching attack, evaluated on (1) $\EmbedScheme_{LSH}$ with $\emblen$ from $8$ to $16$ (round markers), $\flipbias=0$ and (2) $\EmbedScheme_{cPDQ}$ (diamond markers in red box).
    \textbf{Lower-left}: The precision metric $\epsilonPrecR{\epsilonRecall>\epsilonRecallThresh}$ of matching attack, at different recall threshold, with $\emblen=12$ and varying $\flipbias$. 
    The dashed line denotes $50\%$. \textbf{Lower-right}:  Correctness~($\epsilon$, with varying $ \threshold$) and compression efficiency~($\alpha$) tradeoff. ($\emblen=12$, $\flipbias=0.05$)
    Error bars represent the $95\%$ confidence interval.
    }
    
\end{figure*}

\section{SBB to Share Moderation Efforts}
\label{ap:mod}

The SBB mechanism can also be used for social media sites to share moderation efforts across trust boundaries.
Small social media sites have fewer moderation resources hence often have to rely on databases of known bad contents such as the ThreatExchange service~\cite{threatexchange} provided by Facebook,
while preserving the privacy of their users.
In this section, we review the deployment scenario of centralized moderators of social media sites using a third-party database to improve their moderation efforts.
Unlike the deployment scenario where individual users use SBB, the centralized moderator is able to deduplicate the contents that they receive on the forum and reduce the repetition rate of the popular images that are sent in the queries.
Using a dataset of images collected from the imageboard 4chan by Matatov et al.~\cite{matatovdejavu}, 
a social media site that is known for inappropriate content and could use our service, 
we conduct experimental analyses to find suitable parameters of SBB under this setting.
Our results show that while preserving client's privacy and ensuring correctness, SBB compresses the database to $2.1\%$.

Our simulation dataset consists of $1.2$ million posts with images collected from 4chan's politically incorrect board (/pol/) between July 10, 2018 and October 31, 2019,
with $750$ thousand unique PDQHashes.
The repetition rate of popular images are much lower as compared to the deployment scenario of individual users,
as can be observed in \figref{fig:overview-ap}~(upper-left).
Most of the images~($78\%$) don't share the same PDQHash value with others~(first row from bottom, lightest shade).
Each post with an image corresponds to a query for that image to the SBB protocol.
Similar to that in \secref{sec:expr}, we examine the security of SBB against the matching attack of the most popular image,
which appears in $0.02\%$ of all posts.

When comparing the baseline of using $\EmbedScheme_{cPDQ}$ with $\EmbedScheme_{LSH}$,
\figref{fig:overview-ap} (upper-right) shows similar trend as in \figref{fig:privacy}~(left).
First, we notice that the accuracy metric cannot accurately describe the privacy leakage,
while $\epsilonAuc$ remains $100\%$ in all experiments.
Second, across all metrics, $\EmbedScheme_{cPDQ}$ allows the adversary to have perfect improvement when performing the matching attack,
while $\EmbedScheme_{LSH}$ offers better security guarantees.
In particular, when measured by $\epsilonPrecR{\epsilonRecall=100\%}$,
$\EmbedScheme_{LSH}$ performs better than $\EmbedScheme_{cPDQ}$ even when the coarse embeddings are of similar length.
The adversary achieves $89.5\%$ precision upon receiving a $16$-bit LSH-based coarse embedding as request,
as compared to almost $96.6\%$ precision when we apply $\EmbedScheme_{cPDQ}$.
In \figref{fig:overview-ap}~(lower-left), we fix the coarse embedding length at $\emblen=12$ and present the results for $\epsilonPrecR{\epsilonRecall>\epsilonRecallThresh}$ with varying $\flipbias$.
Similar to that in \figref{fig:privacy}~(right), 
for any $\flipbias$, $\epsilonPrecR{\epsilonRecall>\epsilonRecallThresh}$ is smaller than $50\%$~(noted by the horizontal line) for all recall thresholds, satisfying our security requirement as defined in \secref{sec:expr}.

In \figref{fig:overview-ap}~(lower-right), we fix $\emblen=12$ and $\flipbias=0.05$, and vary the value of the coarse threshold $\coarsethreshold$.
Similar to previous experiment, we use the 4chan dataset as $\dist$,
and all the possible values of PDQHash in the dataset as $\imageset$.
The results are aggregated from $2$M requests randomly selected from the dataset with replacement.
The parameter settings are less conservative in this scenario, allowing better efficiency. 
For example,
when $\threshold=32$,
a sufficiently correct embedding scheme with $\epsilon \geq 95\%$~(noted by the dashed horizontal line), 
produces a smaller bucket~($2.1\%$) in \figref{fig:overview-ap}~(lower-right) than that~($9.3\%$) in \figref{fig:accuracy}~(left).

The distribution of images shared by clients 
may differ depending on the deployment context.
Such difference impacts the adversarial advantage when performing matching attacks.
Our framework allows simulation on real-world social media datasets,
in order to find good parameters for $\EmbedScheme_{LSH}$ when navigating the privacy/efficiency tradeoff.

\begin{figure*}
\centering
\begin{adjustbox}{width=\textwidth}
   \begin{tikzpicture}
     \begin{axis}[
            name=plot1,
    	    legend columns = 1,
    	    ticklabel style = {font=\footnotesize},
    	    xmin=0.5,
    	    xmax=5.5,
	        ymin=-0.05,
	        ymax=1.05,
            xticklabel style={align=center},
            yticklabel={\pgfmathparse{\tick*100}\pgfmathprintnumber{\pgfmathresult}\%},
            ylabel={$\epsilonPrecR{\epsilonRecall > \epsilonRecallThresh}$},
            ylabel style={font=\footnotesize, at={(axis description cs:0.08,.5)},anchor=south},
            xlabel={Repetition frequency \repeats},
            xlabel style = {font=\footnotesize}, 
    		height=3.5cm,
    		width=0.41\textwidth,
    		scaled y ticks=false,
    	    legend style={nodes={scale=.95, font=\footnotesize}}, 
    		legend style={/tikz/every even column/.append style={column sep=0.38cm}, at={(1.01,1)},anchor=north west,font=\footnotesize},
    	]

      \addplot[clr5_1,mark=*,line width=1pt,error bars/.cd, y dir=both, y explicit,] coordinates {
(1.0000, 0.4090) += (0, 0.0285) -= (0, 0.0285)
(2.0000, 0.9956) += (0, 0.0020) -= (0, 0.0020)
(3.0000, 1.0000) += (0, 0.0000) -= (0, 0.0000)
(4.0000, 1.0000) += (0, 0.0000) -= (0, 0.0000)
(5.0000, 1.0000) += (0, 0.0000) -= (0, 0.0000)

       };
       \addplot[clr5_2,mark=*,line width=1pt,error bars/.cd, y dir=both, y explicit,] coordinates {
(1.0000, 0.4090) += (0, 0.0285) -= (0, 0.0285)
(2.0000, 0.9956) += (0, 0.0020) -= (0, 0.0020)
(3.0000, 1.0000) += (0, 0.0000) -= (0, 0.0000)
(4.0000, 1.0000) += (0, 0.0000) -= (0, 0.0000)
(5.0000, 1.0000) += (0, 0.0000) -= (0, 0.0000)
       };
       \addplot[clr5_3,mark=*,line width=1pt,error bars/.cd, y dir=both, y explicit,] coordinates {
(1.0000, 0.4090) += (0, 0.0285) -= (0, 0.0285)
(2.0000, 0.9867) += (0, 0.0034) -= (0, 0.0034)
(3.0000, 0.9999) += (0, 0.0001) -= (0, 0.0001)
(4.0000, 1.0000) += (0, 0.0000) -= (0, 0.0000)
(5.0000, 1.0000) += (0, 0.0000) -= (0, 0.0000)

      };
  		
      \addplot[clr5_4,mark=*,line width=1pt,error bars/.cd, y dir=both, y explicit,] coordinates {
(1.0000, 0.1471) += (0, 0.0060) -= (0, 0.0060)
(2.0000, 0.9592) += (0, 0.0065) -= (0, 0.0065)
(3.0000, 0.9997) += (0, 0.0002) -= (0, 0.0002)
(4.0000, 1.0000) += (0, 0.0000) -= (0, 0.0000)
(5.0000, 1.0000) += (0, 0.0000) -= (0, 0.0000)

       };
       \addplot[clr5_5,mark=*,line width=1pt,error bars/.cd, y dir=both, y explicit,] coordinates {
(1.0000, 0.0057) += (0, 0.0013) -= (0, 0.0013)
(2.0000, 0.0657) += (0, 0.0141) -= (0, 0.0141)
(3.0000, 0.5056) += (0, 0.1926) -= (0, 0.1926)
(4.0000, 0.9740) += (0, 0.0177) -= (0, 0.0177)
(5.0000, 1.0000) += (0, 0.0000) -= (0, 0.0000)

       };
       \addplot [black, no markers, line width=1pt,dashed] coordinates {(0.5,0.5) (5.5,0.5)};

       \legend{{$\epsilonRecall > 0\%$}, {$\epsilonRecall > 25\%$},
       {$\epsilonRecall > 50\%$},{$\epsilonRecall > 75\%$}, {$\epsilonRecall =100\%$}}
    	\end{axis}
      \begin{axis}[
            name=plot2,
            at=(plot1.right of south east), anchor=left of south west,
    	    legend columns = 1,
    	    ticklabel style = {font=\footnotesize},
    	    xmin=0.5,
    	    xmax=5.5,
	        ymin=-0.05,
	        ymax=1.05,
            xticklabel style={align=center},
            yticklabel={\pgfmathparse{\tick*100}\pgfmathprintnumber{\pgfmathresult}\%},
            ylabel style={font=\footnotesize, at={(axis description cs:0.08,.5)},anchor=south},
            xlabel={Repetition frequency \repeats},
            xlabel style = {font=\footnotesize}, 
    		height=3.5cm,
    		width=0.41\textwidth,
    		scaled y ticks=false,
    	    legend style={nodes={scale=.95, font=\footnotesize}}, 
    		legend style={/tikz/every even column/.append style={column sep=0.38cm}, at={(1.01,1)},anchor=north west,font=\footnotesize},
    	]

      \addplot[clr5_1,mark=*,line width=1pt,error bars/.cd, y dir=both, y explicit,] coordinates {
(1.0000, 0.3896) += (0, 0.0310) -= (0, 0.0310)
(2.0000, 0.5028) += (0, 0.0225) -= (0, 0.0225)
(3.0000, 0.5290) += (0, 0.0227) -= (0, 0.0227)
(4.0000, 0.4859) += (0, 0.0268) -= (0, 0.0268)
(5.0000, 0.5169) += (0, 0.0254) -= (0, 0.0254)
       };
       \addplot[clr5_2,mark=*,line width=1pt,error bars/.cd, y dir=both, y explicit,] coordinates {
(1.0000, 0.3896) += (0, 0.0310) -= (0, 0.0310)
(2.0000, 0.5028) += (0, 0.0225) -= (0, 0.0225)
(3.0000, 0.5290) += (0, 0.0227) -= (0, 0.0227)
(4.0000, 0.4812) += (0, 0.0246) -= (0, 0.0246)
(5.0000, 0.5107) += (0, 0.0273) -= (0, 0.0273)
       };
       \addplot[clr5_3,mark=*,line width=1pt,error bars/.cd, y dir=both, y explicit,] coordinates {
(1.0000, 0.3896) += (0, 0.0310) -= (0, 0.0310)
(2.0000, 0.4585) += (0, 0.0209) -= (0, 0.0209)
(3.0000, 0.5243) += (0, 0.0236) -= (0, 0.0236)
(4.0000, 0.4796) += (0, 0.0239) -= (0, 0.0239)
(5.0000, 0.5084) += (0, 0.0269) -= (0, 0.0269)
      };
  		
      \addplot[clr5_4,mark=*,line width=1pt,error bars/.cd, y dir=both, y explicit,] coordinates {
(1.0000, 0.1436) += (0, 0.0062) -= (0, 0.0062)
(2.0000, 0.4000) += (0, 0.0167) -= (0, 0.0167)
(3.0000, 0.5207) += (0, 0.0233) -= (0, 0.0233)
(4.0000, 0.4767) += (0, 0.0235) -= (0, 0.0235)
(5.0000, 0.5060) += (0, 0.0280) -= (0, 0.0280)

       };
       \addplot[clr5_5,mark=*,line width=1pt,error bars/.cd, y dir=both, y explicit,] coordinates {
(1.0000, 0.0058) += (0, 0.0013) -= (0, 0.0013)
(2.0000, 0.0512) += (0, 0.0081) -= (0, 0.0081)
(3.0000, 0.1162) += (0, 0.0383) -= (0, 0.0383)
(4.0000, 0.2169) += (0, 0.0593) -= (0, 0.0593)
(5.0000, 0.3851) += (0, 0.0695) -= (0, 0.0695)

       };
       \addplot [black, no markers, line width=1pt,dashed] coordinates {(0.5,0.5) (5.5,0.5)};

       \legend{{$\epsilonRecall > 0\%$}, {$\epsilonRecall > 25\%$},
       {$\epsilonRecall > 50\%$},{$\epsilonRecall > 75\%$}, {$\epsilonRecall =100\%$}}
    	\end{axis}
     \end{tikzpicture}
    \end{adjustbox}
    \vspace{-0.7cm}
    \caption{\label{fig:repeated-ap} The precision metric $\epsilonPrecR{\epsilonRecall>\epsilonRecallThresh}$ of matching attack with randomized~(\textbf{left}) and fixed~(\textbf{right}) indexing functions, at different recall threshold, with $\emblen=9$, $\flipbias=0.05$ and varying repetition rate $\repeats$. 
    The dashed lines denote $50\%$. 
    }
    
\end{figure*}

\section{Simulating  Repeated Queries}
\label{ap:repeated}

We consider the following attack scenario in this section.
An adversary who receives $\repeats$ embeddings and is able to infer that the embeddings correspond to the same image aims to guess if the image has the same PDQHash value as an adversary chosen image $\image_{adv}$.
The embeddings are generated with independent coins.
We first generalize from Theorem~\ref{thm:bayes-opt} to compute the probability used by a Bayes optimal adversary.

\begin{theorem}

Let $\Delta_I(\simimage, p) = \Delta(I(\simimage), p)$
for a similarity hash $\simimage \in \bits^\len$.
Consider a fixed image $\image_{adv} \in \imagespace$.

Suppose we sample an image $\image$ and generate $q$ embeddings for this image as follows:

$$\image \getsr \dist, (I_1, p_1) \getsr \Embed(\image), \ldots, (I_q, p_q) \getsr \Embed(\image)$$.

Let $\simimage_{adv} = \simhash(\image_{adv})$ and $\simimage = \simhash(\image)$. Furthermore, consider the distribution $\distf$ that $\simhash$ induces on $\bits^\len$, where $$\distf(x) = \Pr[\simhash(\image) = x] = \sum\limits_{\substack{\image' \in \imagespace,\\ \simhash(\image') = x}} \dist(\image')$$ for all $x \in \bits^\len$.

We have that 
\begin{equation*}
\begin{multlined}
\Pr[\simimage = \simimage_{adv} \mid (I_1, p_1), \ldots, (I_\repeats, p_\repeats)] \\
= \dfrac{\gamma^{\delta(\simimage_{adv})}\cdot (1 - \gamma)^{\repeats \cdot d - \delta(\simimage_{adv})} \cdot \distf(\image_{adv})}{\sum\limits_{\simimage' \in \bits^\len} \gamma^{\delta(\simimage')}\cdot (1 - \gamma)^{\repeats \cdot d - \delta(\simimage')} \cdot \distf(\image')}
\end{multlined}
\end{equation*}
where $\delta(\simimage) = \sum\limits_{j \in [q]} \Delta_{I_j}(\simimage, p_j)$ and the probability is over the sampling of $\image$ and the coins used in the invocations of $\Embed$.
\end{theorem}
\begin{proof}
The proof follows that of Theorem~\ref{thm:bayes-opt} except we now consider independent bit flips over multiple embeddings corresponding the same image $\image$, hence the sum in the exponents of $\gamma$ and $1 - \gamma$.
\end{proof}

We experiment with the deployment scenario of when individual users use our services and follow the experimental setup as in \secref{sec:data} with the same dataset.
\figref{fig:repeated-ap} presents the adversary advantage as measured by the precision metric $\epsilonPrecR{\epsilonRecall>\epsilonRecallThresh}$~(Y axes) with the increase of the repetition frequency $\repeats$~(X axes).
The left figure shows the drastically increasing adversary advantage when embeddings are randomized.
Even with low repetition rates, e.g. $\repeats=2$, $\epsilonPrecR{\epsilonRecall>0\%}$ is almost $100\%$.
This is because for two randomized queries of the same image,
the amount of information leaked to the adversary is around $2\emblen$ bits when the repetition rate of picking the same indexing function is low.
Hence the privacy loss of $\emblen=9$ and $\repeats=2$ would be similar to that of $\emblen=18$,
which is detrimental.
However, when the embeddings are derandomized,
i.e. in the case when indexing functions are fixed for repeated queries, 
the adversary advantage is at most around $50\%$ as shown in the right figure, even for $\epsilonPrecR{\epsilonRecall>0\%}$ and $\repeats=5$.
This is because the adversary doesn't receive much extra information from the repeated queries, nevertheless, $\epsilonPrecR{\epsilonRecall=100\%}$ keeps increasing with the increase of $\repeats$,
as the adversary becomes more certain when predicting more cases of $\image$ as positive.

\end{document}